\RequirePackage[l2tabu, orthodox]{nag}		
\documentclass[10pt]{article}

\usepackage[T1]{fontenc}
\usepackage{lmodern}
\usepackage[frenchmath]{mathastext}							
\usepackage{natbib}       									
\usepackage{amsmath}                        
\numberwithin{equation}{section}
\usepackage{graphicx} 											
\usepackage{adjustbox,float}
\usepackage{enumitem}												
\usepackage{mdwlist}												
\usepackage[dvipsnames]{xcolor}							
\usepackage[plainpages=false, pdfpagelabels]{hyperref} 
	\hypersetup{
		colorlinks   = true,
		citecolor    = RoyalBlue, 
		linkcolor    = RubineRed, 
		urlcolor     = Turquoise
	}
\usepackage{amssymb}                        
\usepackage[mathscr]{eucal}                 
\usepackage{dsfont}													
\usepackage[paperwidth=8.5in,paperheight=11in,top=1.25in, bottom=1.25in, left=1.0in, right=1.0in]{geometry} 
\usepackage{mathtools}                      
\mathtoolsset{showonlyrefs=true}          
\usepackage{amsthm}                         
	\allowdisplaybreaks                       
	\theoremstyle{plain}
	\newtheorem{theorem}{Theorem}
	\numberwithin{theorem}{section}
	\newtheorem{lemma}[theorem]{Lemma}       	
	\newtheorem{proposition}[theorem]{Proposition}
	\newtheorem{corollary}[theorem]{Corollary}
	\theoremstyle{definition}
	\newtheorem{definition}[theorem]{Definition}
	
	\newtheorem{remark}[theorem]{Remark}

\newif\ifcomment

\commentfalse
\commenttrue 

\ifcomment
	\newcommand{\mcomment}[1]{\marginpar{\tiny\textbf{\color{blue} #1}}}
\else
	\newcommand{\mcomment}[1]{}
\fi

\newcommand{\E}{\mathbb{E}\thinspace }

\newcommand \vp {v}

\newcommand\Eb{\mathds{E}}

\newcommand\Pb{\mathds{P}}
\newcommand\Qb{\mathds{Q}}
\newcommand\Rb{\mathds{R}}

\newcommand\Ac{\mathscr{A}}

\newcommand\Dc{\mathscr{D}}

\newcommand\Fc{\mathscr{F}}

\newcommand\Lc{\mathscr{L}}

\newcommand\eps{\varepsilon}

\newcommand\sig{\sigma}
\newcommand\Sig{\Sigma}

\newcommand\gam{\gamma}

\newcommand\del{\delta}
\newcommand\Del{\Delta}
\newcommand\kap{\kappa}

\newcommand\gamo{\overline{\gamma}}

\newcommand\fv{\mathbf{f}}

\newcommand\Wv{\textbf{W}}
\newcommand\alv{{\boldsymbol\alpha}}
\newcommand\betav{{\boldsymbol\beta}}

\newcommand\lamv{{\boldsymbol\lambda}}

\newcommand\piv{{\boldsymbol\pi}}

\newcommand\Wt{\widetilde{W}}

\newcommand\St{\widetilde{S}}
\newcommand\Tt{\widetilde{T}}

\newcommand\Yt{\widetilde{Y}}

\newcommand\Ut{\widetilde{U}}

\newcommand\kapt{\widetilde{\kappa}}

\newcommand\alvt{\widetilde{\alv}}
\newcommand\betavt{\widetilde{\betav}}

\newcommand\pit{\widetilde{\pi}}
\newcommand\mut{\widetilde{\mu}}
\newcommand\pivt{\widetilde{\piv}}

\newcommand\dd{\mathrm{d}}
\newcommand\ee{\mathrm{e}}

\newcommand\Cov{\operatorname{Cov}}
\newcommand\Var{\operatorname{Var}}

\newcommand\altt{\varrho}
\providecommand{\keywords}[1]{\textbf{\textit{Keywords: }} #1}

\begin{document}

\title{Optimal Dynamic Basis Trading}

\author{
B. Angoshtari
\thanks{Applied Mathematics Department,  University of Washington, Seattle WA 98195.  \textbf{e-mail}: \url{bahmang@uw.edu}}
\and
T. Leung
\thanks{Applied Mathematics Department,  University of Washington, Seattle WA 98195.  \textbf{e-mail}: \url{timleung@uw.edu}}
}

\date{This version: \today}

\maketitle

\begin{abstract}
We study the problem of dynamically trading a \mbox{futures contract} and its underlying asset under a stochastic basis model.  {The basis evolution is modeled by a stopped scaled Brownian bridge to account for non-convergence of the basis at maturity.} The optimal trading strategies are determined  from  a utility maximization problem under hyperbolic absolute risk aversion (HARA) risk preferences. By analyzing the associated Hamilton-Jacobi-Bellman equation, we derive the exact conditions under which the \mbox{equation} admits  a solution and   solve the utility maximization explicitly.  A series of numerical examples are provided to illustrate the optimal strategies and examine the effects of model parameters.
\end{abstract}

\keywords{futures \and stochastic basis \and cash and carry \and scaled Brownian bridge \and risk aversion} 

\noindent \textbf{JEL Classification}  C41 G11 G12
%
%
%
%
%
%
%


%
%

\section{Introduction}

\emph{Basis trading}, also known as \emph{cash-and-carry} trading in the context of futures contracts, is a core   strategy for many speculative traders  who seek to profit from anticipated convergence of spot and futures prices. The practice usually involves taking a long position in the under-priced asset and a short position in the over-priced one, and closing the positions when convergence occurs. In reality, however, basis trading is far from a riskless arbitrage. Unexpected changes in market factors such as interest rate, cost of carry, or dividends can diminish profitability. Moreover, market frictions, such as transaction costs and collateral payments, can turn seemingly certain arbitrage opportunities into disastrous trades. It is also possible that the basis does not converge at maturity. This non-convergence phenomenon  was commonly observed in the grains markets. As reported in \cite{irwin2011}, \cite{adjemian2013}, and \cite{Garcia2015}, for most of 2005-2010  futures contracts expired up to 35\% above the spot price. As a result, some cash-and-carry traders may choose to close their positions prior to maturity to limit risk exposure. 

Early works on pricing of futures contracts, such as \cite{CoxIngersollRoss1981} and \cite{Modest1983}, established no-arbitrage relationships between the spot price and associated futures prices.  Assuming imperfections, such as transaction cost, these relationships take the form of pricing bounds, which can be used for identifying profitable trades. Among related  studies on \mbox{basis} trading,  \cite{BrennanSchwartz1988} and \cite{BrennanSchwartz1990}  assumed that the basis of an index futures follows a scaled Brownian bridge and that trading the assets is subject to position limits and transaction costs. They calculated the value of the embedded timing options to trade the basis, and used the option prices to devise open-hold-close strategies involving the index futures and the underlying index.  Also under a Brownian bridge model,   \cite{dai2011optimal} provided  an alternative strategy and specification of transaction costs. Another related work by \cite{LiuLongstaff2004} assumed that the basis follows a scaled Brownian bridge and   the investor is subject to a collateral constraint. They derived the  closed-form   strategy that maximizes  the expected logarithmic utility of terminal wealth, and showed the optimality of taking smaller arbitrage positions well within  the collateral constraint.

{
In the aforementioned studies on optimal basis trading, the market model contains arbitrage. Indeed, it is assumed that the basis, which is a tradable asset, converges to zero at a fixed future time. In this paper, we consider a different scenario where the basis does not vanish    at maturity. More precisely, we model the stochastic  basis by a scaled Brownian bridge that is stopped before its achieves convergence. Our proposed model is motivated by the market phenomenon of non-convergence as well as  the possibility of  traders  closing their futures positions prior to maturity. It also has the added advantage that it generates an arbitrage-free futures market (see  Proposition \ref{prop:NA} below).
}

{
We consider a general class of risk preferences by using hyperbolic absolute risk aversion (HARA) utility functions which includes power (CRRA) and exponential (CARA) utilities. We exclude, however, the case of decreasing risk-tolerance functions and, in particular, the quadratic utility function.}
Among our findings, we derive in closed form the optimal dynamic basis trading strategy maximizing the expected utility of terminal wealth.  In solving our portfolio optimization problem, the critical question of well-posedness arises. To that end, we find the exact conditions under which the maximum expected utility is finite. For the case where the expected utility explodes, we derive the critical investment horizon at which the explosion happens. See Section \ref{sec:HJB} for further details. Note that achieving infinite expected utility has been observed in the following contexts: infinite horizon portfolio  optimization (\cite{Merton1969}), optimal execution (\cite{bulthuis2016optimal}), and finite horizon optimal trading of assets with mean-reverting return (\cite{KimOmberg1996};\cite{KornKraft2004}). The latter studies, respectively, have coined the terms ``\emph{nirvana strategies}'' and ``\emph{I-unstable}'' for investment strategies that yield infinite expected utility in finite investment horizon.

Our model is related to a number of studies in finance  involving Brownian bridges. Applications include  modeling the flow of information in the market. For example, \citet*{brody2008information} used a Brownian bridge as the noise in the information about a future market event, and derived option pricing formulae based on this asset price dynamics and   market information flow. \citet*{cartea2016algorithmic2} utilized a randomized Brownian bridge (rBb) to model the mid-price of an asset with a random end-point perceived by an informed trader, and   determined the optimal placements  of market and limit orders.  \citet*{leung2017optimal} also applied a rBb model to investigate the optimal timing to sell different option positions. 

Basis trading involves trading a single futures contract along with its spot asset. A similar alternative strategy is to trade multiple futures with the spot but different maturities. For instance, \cite{LeungYan2018,LeungYan2019} considered futures prices generated from the same two-factor stochastic spot model and optimize dynamic trading strategies.  In optimal convergence trading, asset prices or their spreads are often modeled by a stationary mean-reverting process. See, for example, \cite{KimOmberg1996}, \cite{KornKraft2004}, \cite{Primbsetal2008}, \cite{ChiuWong2011}, \cite{LiuTimmermann2012}, \cite{tourin2013dynamic}, \cite{cartea2016algorithmic}, \cite{lee2016pairs}, \cite{LeungXin2016}, \cite{KitapLeung2017}, and  \cite{CarteaGanJaimungal2018}. In these studies, however, prices or spreads are not scheduled to converge at a  future time. Using a scaled Brownian bridge, we can control  the basis's tendency to converge towards the expiration date.

The rest of the paper is organized as follows. In Section \ref{sec:model}, we introduce our market model and formulate the dynamic basis trading problem.  In Section \ref{sec:HJB}, we solve the associated Hamilton-Jacobi-Bellman partial differential equation. Section \ref{sec:Opt} contains  our results on the  optimal basis trading strategy, along with a series of illustrative   numerical examples. Section \ref{sect-conclude} concludes.  Longer proofs are included in the Appendix.


%
%
\section{Problem setup}\label{sec:model}
We consider an investor who trades a riskless asset, a futures contract $F$ with maturity $T$ and its underlying asset $S$, over a period $[0,T]$. For simplicity, we assume that $S$ does not pay dividends and has no storage cost. We also assume that the interest rate is zero, by taking the riskless asset as the numeraire. Under these assumptions, the futures price, the forward price, and the spot price should be  equal. In practice, however, market frictions and inefficiencies may render  futures price   different from the spot or forward price. As discussed above, the futures price may not even converge to the spot price at maturity.

Motivated by these market  imperfections, we propose a stochastic model that incorporates the comovements of  the futures and spot prices,  and   captures the tendency of the basis to approach zero but not necessarily  vanish at maturity. In essence, the spot and futures prices are drive by correlated Brownian motions and  the stochastic basis is represented by scaled Brownian bridge, as we will show in Lemma \ref{lem:Bb} below.

To describe our model, we assume that the volatility normalized\footnote{See Remark \ref{rem:VolNorm}.} futures price $(F_t)_{t\in[0,T]}$ and the volatility normalized spot price   $(S_t)_{t\in[0,T]}$ satisfy
\begin{align}\label{eq:S}
	\frac{\dd S_t}{S_t} &= \mu_1 \dd t + \dd W_{t,1},
\intertext{and}
	\label{eq:F}
	\frac{\dd F_t}{F_t} &=
	\left(\mu_2 + \frac{\kap\,Z_t}{T-t+\eps}\right)\dd t
	+ \rho \,\dd W_{t,1} + \sqrt{1-\rho^2}\,\dd W_{t,2},
\intertext{where  $(Z_t)_{0\le t\le T}$ is the log-value of the stochastic \emph{basis} defined by}
	\label{eq:Z}
	Z_t &:= \log\left(\frac{S_t}{F_t}\right);\quad 0\le t\le T.
\end{align}
Here, $\Wv_t=(W_{t,1},W_{t,2})^\top$ is a standard Brownian motion in a filtered probability space $\big(\Omega,\Fc,\Pb,(\Fc_t)_{t\ge0}\big)$ where $(\Fc_t)_{t\ge0}$ is generated by the Brownian motion and satisfies the usual conditions. The parameters, $\mu_1$ and $\mu_2$, with $\mu_1, \mu_2\in \mathbb{R}$,  represent the Sharpe ratios of $S$ and $F$ respectively (see Remark \ref{rem:VolNorm} below). As specified in \eqref{eq:F}, the futures price has the tendency to revert around $S$. Indeed, if $F_t$ is significantly higher than $S_t$, then $Z_t$ becomes negative. Consequently, the drift of $F_t$ can also be negative, and thus driving the value of $F_t$ downward to be closer to $S_t$. The opposite will hold if $F_t$ is significantly lower  than $S_t$. The  speed of this mean reversion  is reflected  by the constant $\kap > 0$.  On the other hand, the two prices need not coincide at maturity. The level of non-convergence is  controlled by the parameter $\eps > 0$. A smaller $\eps$ means that at maturity  $S_T$ and $F_T$ tend to be closer. In fact, if $\eps=0$,   $S_T$ and $F_T$ will be exactly the same. Lastly, we incorporate correlation between the two price processes through the parameter  $|\rho|<1$.  


\begin{remark}\label{rem:VolNorm}
	Let $(\St_t)$ satisfying
	\begin{align}
		\dd \St_t = \St_t\left(\mu_t \dd t + \sig_t \dd W_t\right),
	\end{align}
	be the quoted price of an asset. The \emph{``volatility normalized price''} of the asset $(S_t)$ is given by
	\begin{align}
		\dd S_t = \frac{S_t}{\sig_t\St_t}\, \dd \St_t = S_t\left(\frac{\mu_t}{\sig_t} \dd t + \dd W_t\right),\quad S_0=\St_0,
	\end{align}
	such that the volatility of $(S_t)$ is 1.
	
	Furthermore, let $(\pit_t)$ be the amount invested in the asset. The \emph{``volatility adjusted position''} in asset $S$ is given by $\pi_t := \sig_t \pit_t$ such that
	\begin{align}
		\dd \pit_t = \pit_t \frac{\dd \St_t}{\St_t} = \pi_t \frac{\dd S_t}{S_t}.
	\end{align}
	In other words, when working with volatility normalized prices,  volatility adjusted positions  should be  used.  To find the \$ positions, we only need to divide the volatility adjusted positions by   volatility, i.e. $\pit_t = \frac{\pi_t}{\sig_t}$. 
\end{remark}

	For the rest of the article, we use the terms prices and positions in place of volatility normalized prices and volatility adjusted positions.

As we show in the following lemma, the price dynamics \eqref{eq:S} and \eqref{eq:F} imply that the stochastic basis $(Z_t)$ is a scaled Brownian bridge that converges to zero at $T+\eps$. This implies  that $S$ and $F$ converge at $T+\eps$, i.e. $\lim_{t\to T+\eps}(F_t / S_t) = 1$, $\Pb$-almost surely. However, since the futures contract expires at $T$ and $T+\eps$ is past maturity,  such a  convergence is not realized in the market. In a limiting case of our model where $\eps\to0^+$, the spot and future prices converge at $T$ and the market model admits arbitrage. In contrast, for any $\eps>0$,   the market is arbitrage-free. See Proposition \ref{prop:NA} and Remark \ref{rem:arb}  below.

\begin{lemma}\label{lem:Bb}
	The stochastic basis $(Z_t)_{0\le t\le T}$ satisfies 
	\begin{align}\label{eq:Bb}
		\dd Z_t &= \mu_1 - \mu_2 -\frac{\kap\, Z_t}{T-t+\eps}\dd t + (1-\rho)\dd W_{t,1} - \sqrt{1-\rho^2}\, \dd W_{t,2},		
	\end{align}
	for $0\le t\le T$. In particular, if we consider the solution of this SDE over $[0,T+\eps]$, then $Z_{T+\eps} = 0$, $\Pb$-almost surely.
\end{lemma}
\begin{proof}
	We obtain  \eqref{eq:Bb} by applying \eqref{eq:S}-\eqref{eq:Z} to  $\dd Z_t = \frac{\dd S_t}{S_t} - \frac{\dd F_t}{F_t}$. By Eq. (6.6) on page 354 of \cite{KartzasShreve1991}, the unique strong solution of \eqref{eq:Bb} is
	\begin{align}\label{eq:Zsol}
	\begin{split}
		Z_t = &Z_0 \left(1-\frac{t}{T+\eps}\right)^\kap + (\mu_1-\mu_2) A(t)\\
		&+ \int_0^t \left(\frac{T-t+\eps}{T-u+\eps}\right)^\kap \left((1-\rho)\dd W_{u,1} - \sqrt{1-\rho^2}\, \dd W_{u,2}\right);\quad 0\le t\le T,
	\end{split}
	\end{align}
	where we have defined
	\begin{align}\label{eq:a}
		A(t) :=
		\begin{cases}
			\displaystyle
			\frac{1}{\kap-1}\left(T-t+\eps-\frac{(T-t+\eps)^\kap}{(T+\eps)^{\kap-1}}\right);
			&\quad \text{if } \kap\ne1,\\
			\\
			\displaystyle
			(T-t+\eps)\log\left(\frac{T+\eps}{T-t+\eps}\right);
			&\quad \text{if } \kap=1.
		\end{cases}		
	\end{align}
	Since $\kap>0$, taking limits in \eqref{eq:a} yields  $\lim_{t\to T+\eps} A(t) = 0$. From \eqref{eq:Zsol}, it   follows that $\Pb(Z_{T+\eps}=0)=1$.\qed 
\end{proof}
\medskip

\begin{remark}
	Lemma \ref{lem:Bb} provides an alternative representation of our model. Indeed, one can define the underlying asset $S$ by \eqref{eq:S} and the basis $Z$ by \eqref{eq:Bb}. It will then follows that the futures price defined by $F_t:=S_t \ee^{-Z_t}$ satisfies \eqref{eq:F}.\qed
\end{remark}

\medskip

As a corollary, we now describe  the distribution of the basis $(Z_t)_{t\ge0}$.

\begin{corollary}\label{cor:Gaussian}
	Assume that $Z_0$ is deterministic. Then, the basis $(Z_t)_{0\le t\le T}$ is a Gauss-Markov process with mean function 
	\begin{align}\label{eq:mean}
		m(t) &:= \Eb(Z_t) = Z_0 \left(1-\frac{t}{T+\eps}\right)^\kap 
		+ (\mu_1-\mu_2) A(t),
	\intertext{and covariance function}
		\sig(s,t)&:= \Cov(Z_s,Z_t) =
		\begin{cases}\displaystyle
			\frac{2(1-\rho)}{2\kap-1}\frac{(T-t+\eps)^\kap}{(T-s+\eps)^{\kap-1}} \\
			\qquad\left(1-\left(1-\frac{s}{T+\eps}\right)^{2\kap-1}\right); & \text{if }\kap\ne\frac{1}{2},
			\vspace{1em}\\
			\displaystyle
			2(1-\rho)\sqrt{(T-t+\eps)(T-s+\eps)} \\
			\qquad\log\left(\frac{T+\eps}{T-s+\eps}\right); & \text{if }\kap=\frac{1}{2},
		\end{cases}
		\label{eq:Cov}
	\end{align}
	for all $0< s\le t<T$. Here, $A(t)$ is given by \eqref{eq:a}. 
\end{corollary}
\begin{proof}
	This result   follows from direct calculations using  \eqref{eq:Zsol}.	\qed
\end{proof}

\begin{figure}[t]
\centerline{
\adjustbox{trim={0.05\width} {0.05\height} {0.09\width} {0.11\height},clip}
{\includegraphics[scale=0.3, page=1]{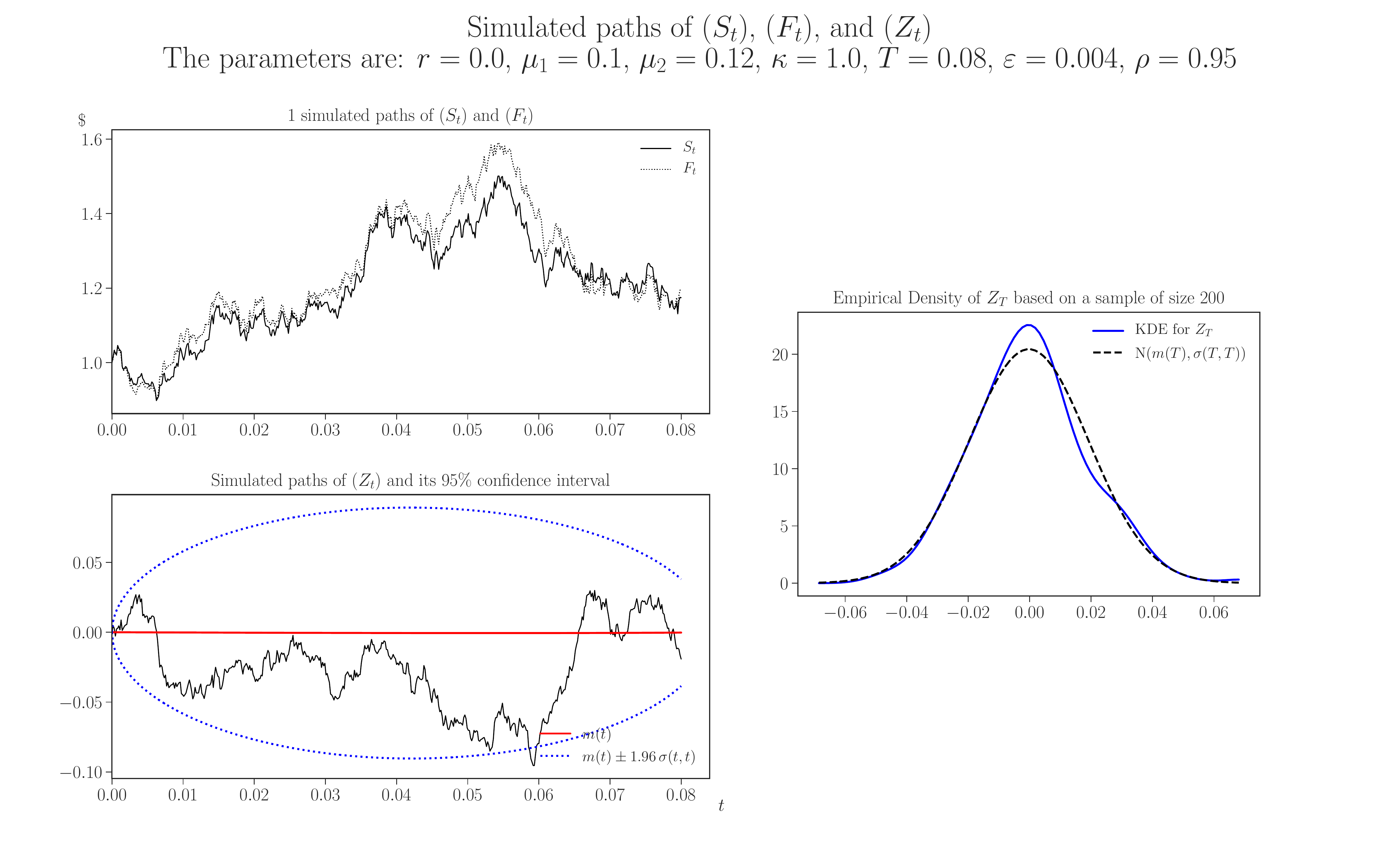}}
}
\caption{The left panel shows the simulated paths of $S$ and $F$ (top), as well as  $Z$ and its confidence interval over time (bottom). On the right, the theoretical and estimated (based on simulated paths) densities of $Z_T$  are plotted. Parameters: $\mu_1=0.1$, $\mu_2=0.12$, $\kap=1$, $T=0.08$, $\eps=0.004$, $\rho=0.95$.
\label{eq:SFZ_sim}}
\end{figure}

From Corollary \ref{cor:Gaussian}, it follows that $Z_T\sim N\left(m(T),\sig(T,T)\right)$. Furthermore, as $\eps\to0$, we have $\Eb(Z_T) = m(T)\to0$ and $\Var(Z_T) = \sig(T,T)\to 0$.  Using Lemma \ref{lem:Bb} and its corollary, it is straightforward to devise a time-discretization scheme to simulate the paths of  $S$, $F$, and $Z$, as shown in Figure \ref{eq:SFZ_sim}. As a confirmation,  we see the empirical density of $Z_T$ based on 200 simulated paths closely matching with the theoretical density, i.e. $N\big(m(T), \sig(T,T)\big)$. 


Next, we prove that our market model is arbitrage-free. Let us first define the \emph{market price of risk} function
\begin{align}\label{eq:lamb}
	\lamv(t,z) &:= \Sig^{-1}
	\begin{pmatrix}
		\mu_1\\
		\mu_2 + \frac{\kap\,z}{T-t+\eps}
	\end{pmatrix}
	=
	\begin{pmatrix}
		\mu_1\\
		\frac{1}{\sqrt{1-\rho^2}}\left(\mu_2-\rho\mu_1 + \frac{\kap}{T-t+\eps}\,z\right)
	\end{pmatrix},
\end{align}
for $0\le t\le T$, where
\begin{align}\label{eq:sig}
	\Sig := 
	\begin{pmatrix}
		1 &0\\
		\rho &\sqrt{1-\rho^2}
	\end{pmatrix}.
\end{align}

\begin{proposition}\label{prop:NA}
	The market model in \eqref{eq:S}--\eqref{eq:Z} is arbitrage-free. In particular, the risk-neutral measure $\Qb$ is given by the Radon-Nikodym derivative $\frac{\dd\Qb}{\dd \Pb} = Y_T$, where the process $(Y_t)_{t\in[0,T]}$ satisfies
\begin{align}\label{eq:MPR}
\begin{cases}\displaystyle
	\frac{\dd Y_t}{Y_t} = -\lamv(t,Z_t)^\top \dd \Wv_t;\quad 0\le t\le T,\\
	Y_0=1,
\end{cases}
\end{align}
and is a $\Pb$-martingale.
\end{proposition}
\begin{proof}
	See Appendix \ref{app:NA}. \qed
	%
\end{proof}

	 Proposition \ref{prop:NA} shows the importance of having $\eps>0$. For $\eps=0$, the market model \eqref{eq:S}-\eqref{eq:F} is not arbitrage-free. Indeed, setting $\eps=0$ in the proof of Lemma \ref{lem:Bb} yields that $\Pb(S_T=F_T)=\Pb(Z_T=0)=1$. Therefore, if there is a risk-neutral measure $\Qb$, then we must have
	\begin{align}
		S_t = \Eb^\Qb(S_T) = \Eb^\Qb(F_T) = F_t;\quad 0\le t\le T,
	\end{align}
	which contradicts \eqref{eq:S}-\eqref{eq:F}.
	
	\begin{remark}\label{rem:arb}
With reference to the proof of Proposition \ref{prop:NA} in Appendix \ref{app:NA}, we can point out  which part would fail if $\eps=0$. Specifically, for $\eps=0$, the key inequality \eqref{eq:Fail} does not hold in the proof.
\end{remark} 

We now discuss the dynamic trading problem faced by the investor. Let $\pit_{t,1}$ be the cash amount  invested in $S$, and   $\pit_{t,2}$  the notional value\footnote{That is, the number of futures contracts held multiplied by the futures price.} invested in $F$, for $t\in[0,T]$. Define the volatility adjusted positions\footnote{See Remark \ref{rem:VolNorm}.} by $\pi_{t,i}:= \sig_{t,i} \pit_{t,i}$, $i\in\{1,2\}$, where $(\sig_{1,t})$ and $(\sig_{2,t})$ are the volatilities of the spot and futures prices, respectively. Then, the trading  wealth $(X_t)_{0\le t\le T}$ follows
\begin{align}\label{eq:Budget}
	\dd X_t &= \pi_{t,1} \frac{\dd S_t}{S_t} + \pi_{t,2} \frac{\dd F_t}{F_t}\\
	&=  \left(\mu_1\, \pi_{t,1}
		  + \left(\mu_2 + \frac{\kap\,Z_t}{T-t+\eps}\right)\pi_{t,2} \right)\dd t\\
	&\qquad + (\pi_{t,1}+\rho\,\pi_{t,2}) \dd W_{t,1}
		   + \pi_{t,2}\sqrt{1-\rho^2} \dd W_{t,2},	
\end{align}
for $0\le t\le T$, and with $X_0=x$.

Next, we define the set of  \emph{admissible trading strategies}.

\begin{definition}\label{def:Admiss}
For constants $x^*, \gam\ge0$, define the set $\Dc\subseteq\Rb$ as follows
\begin{align}
	\Dc :=
	\begin{cases}
		\{x\in\Rb: x>x^*\};&\quad \gam>0,\\
		\Rb;&\quad \gam=0.
	\end{cases}
\end{align}
We denote by $\Ac=\Ac(x^*,\gam)$ the set of all $(\Fc_t)$-adapted processes, denoted by  $\piv=(\pi_{t,1},\pi_{t,2})_{0\le t\le T}$, such that
\begin{enumerate}
	\item[(i)] $\int_0^T \left(\pi_{t,1}^2 + \pi_{t,2}^2 + |\pi_{t,2} Z_t|\right)\dd t < \infty$, $\Pb$-a.s.,
	\item[(ii)] $X_t\in\Dc$ $\Pb$-a.s. for all $t\in[0,T]$, where $(X_t)_{0\le t\le T}$ is given by \eqref{eq:Budget},
	\item[(iii)] $(X_t)_{0\le t\le T}$ is uniformly bounded from below, $\Pb$-a.s. 
\end{enumerate}
\end{definition}

	Note that for $\gam>0$, Condition (ii) becomes $X_t>x^*$ $\Pb$-a.s. for all $t\in[0,T]$, which makes Condition (iii) redundant. For $\gam=0$, Condition (ii) is redundant since the corresponding utility  imposes no constraint on the wealth process. Therefore, for $\gam=0$, Condition (iii) is needed to exclude doubling strategies.



In order to maximize the expected utility of terminal wealth, the investor solves the stochastic control problem
\begin{align}\label{eq:VF}
	V(t,x,z) := \sup_{\piv\in\Ac} \Eb_{t,x,z} \,U\left(X_T\right);\quad (t,x,z)\in[0,T]\times\Dc\times\Rb,
\end{align}
where $\Eb_{t,x,z}(\cdot):=\Eb(\cdot|X_t=x, Z_t=z)$. Here, $U:\Dc\to\Rb$ is a  hyperbolic absolute risk aversion (HARA) utility function whose  \emph{risk tolerance} function admits the  form:
\begin{align}\label{eq:HARA}
	\del(x) := -\frac{U'(x)}{U''(x)}=
	\begin{cases}
		\gam(x-x^*);&\quad \gam>0,\\
		\del_0>0;&\quad \gam=0.
	\end{cases}
\end{align}
We furthermore assume that, for $\gam\ne1$,  
\begin{align}\label{eq:HARA2}
	\frac{U(x)}{U'(x)} = \frac{\del(x)}{\gam-1}=
	\begin{cases}
		\frac{\gam}{\gam-1}(x-x^*);&\quad \gam\in(0,1)\cup(1,+\infty),\\
		-\del_0;&\quad \gam=0.
	\end{cases}
\end{align}

Note that any HARA utility, except the logarithmic utility (i.e. $\gam=1$), can be shifted by a constant such that  \eqref{eq:HARA2} is satisfied. This result is trivial for $\gam=0$ (i.e. exponential utility). The following lemma makes precise  this observation for $\gam\in(0,1)\cup(1,+\infty)$. Hence, condition \eqref{eq:HARA2} results in no loss in generality.
\begin{lemma}
	Let $\Ut:\Dc\to\Rb$ satisfy \eqref{eq:HARA} with $\gam\in(0,1)\cup(1,+\infty)$. Take an arbitrary $x_0\in\Dc$ and define
	\begin{align}\label{eq:HARA2-cond}
		U(x) := \Ut(x)-\Ut(x_0) + \frac{\del(x_0)}{\gam-1}\Ut'(x_0);\quad x\in\Dc.
	\end{align}
	Then, $(\gam-1)U(x)=\del(x)U'(x)$, for $x\in\Dc$.
\end{lemma}
\begin{proof}
	Since $U$ also satisfies \eqref{eq:HARA}, integration by parts yields
	\begin{align}
		U(x)-U(x_0) &= -\int_{x_0}^x \del(y)U''(y)\dd y\\
		&=-\del(x)U'(x) + \del(x_0)U'(x_0)+\gam\big(U(x)-U(x_0)\big).
	\end{align}
	for all $x\in\Dc$. Therefore,
	\begin{align}
		(\gam-1) U(x) =\del(x)U'(x) + (\gam-1)U(x_0) -\del(x_0)U'(x_0),
	\end{align}
	and the result follows since, by \eqref{eq:HARA2-cond}, we have $(\gam-1)U(x_0) =\del(x_0)U'(x_0)$.\qed
\end{proof}

Figure \ref{fig:HARA} illustrates examples of HARA utility functions for different values of $\gam$. Specifically, $\gam\in(0,1)\cup(1,\infty)$ corresponds to power utilities,   $\gam=1$ leads to the logarithmic utility, $\gam=0$ yields the exponential utility, and $\gam=-1$ yields the quadratic utility. Note that we have excluded from our model the case of HARA utility functions with decreasing risk tolerance function, such as quadratic utility.

\begin{figure}[t!]
\centerline{
\adjustbox{trim={0.0\width} {0.04\height} {0.13\width} {0.0\height},clip}
{\includegraphics[scale=0.3, page=1]{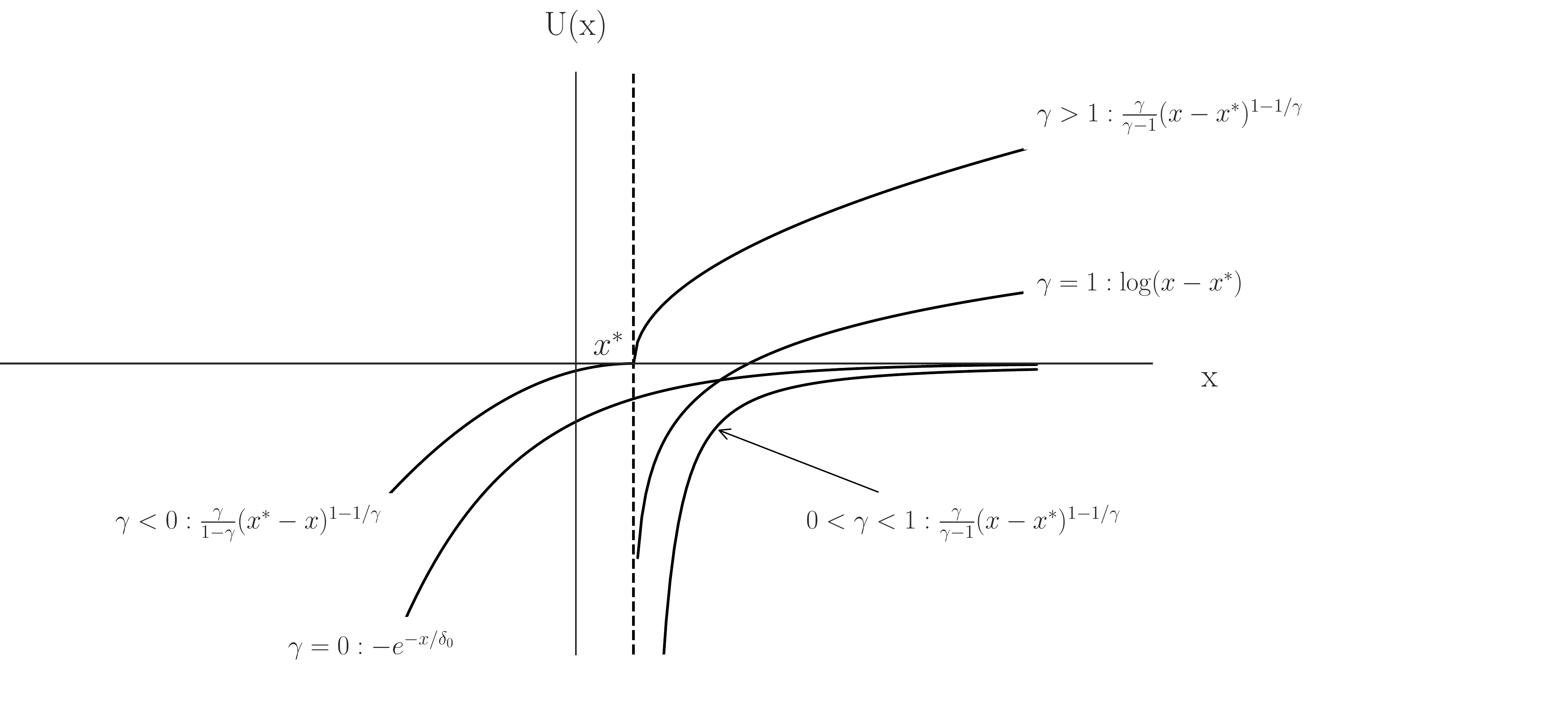}}
}
\caption{HARA utility functions for different values of $\gam$. Specifically, $\gam\in(0,1)\cup(1,\infty)$ yields power utilities, $\gam=1$ yields the logarithmic utility, $\gam=0$ yields the exponential utility, and $\gam=-1$ yields the quadratic utility. Note that we exclude the case of $\gam<0$ from our analysis. 
\label{fig:HARA}}
\end{figure}

\section{HJB equation}\label{sec:HJB}

As we will show in Section \ref{sec:Opt}, the value function $V$ in \eqref{eq:VF} coincides with the classical solution $v(t,x,z)$ of the following terminal value problem on $[0,T]\times \Dc \times \Rb$:
\begin{align}\label{eq:HJB}
\begin{cases}\displaystyle
	\vp_t + \left(\mu_1-\mu_2-\frac{\kap\,z}{T-t+\eps}\right)\, \vp_z 
	+ (1-\rho)\vp_{zz} \displaystyle+\sup_{\piv\in\Rb^2} \Lc_{\piv} \vp = 0,\\
	\vp(T,x,z) = U(x),
\end{cases}
\end{align}
where the differential operator $\Lc_\piv$ is given by
\begin{align}\label{eq:Opr}
\begin{split}
	\Lc_\piv \varphi(x,z) :=
	&\left(\mu_1\, \pi_1 + \left(\mu_2 + \frac{\kap\,z}{T-t+\eps}\right)\pi_2 \right)\,\varphi_x\\
	&+\frac{1}{2}\left(\pi_1^2 + \pi_2^2 + 2\pi_1\pi_2\rho\right)\,\varphi_{xx}
	+(1-\rho)(\pi_1-\pi_2) \,\varphi_{xz},
\end{split}
\end{align}
for any $\piv\in\Rb^2$ and any function $\varphi(t,x,z)$ with continuous derivatives $\varphi_{xx}$ and $\varphi_{xz}$.

\subsection{Well-posedness Conditions}
For the rest of this section, we derive the solution for the nonlinear Hamilton-Jacobi-Bellman (HJB) equation \eqref{eq:HJB}. As it turns out, this equation does not admit  a solution for \emph{all} parameter  values. This leads us to determine and study  the exact conditions under which \eqref{eq:HJB} has a solution. To prepare for our results, we define the constant
\begin{align}\label{eq:gamo}	
	\gamo := \frac{\sqrt{2 \kap^2 + (1 - 2 \kap)(1+\rho)}}{\sqrt{2} \kap (1-\rho )}
	-\frac{1 - 2 \kap +\rho}{2 \kap (1-\rho)}\ge1.
\end{align}
Note that $2 \kap^2 + (1 - 2 \kap)(1+\rho)=(\kap-1)^2+(\kap-\rho)^2+\rho(1-\rho)>0$, and $\gamo\ge1$ since we have assumed that $\kap>0$ and $|\rho|<1$. We now identify  the following cases: 
\vspace{5pt}
\begin{enumerate}
	\item[(i)] For $\gam> \gamo$, \eqref{eq:HJB} is ``\emph{ill-posed}'' in the sense that it has a solution only if $T<T^*(\gam)$, where $T^*(\gam)$ is given below.
	\item[(ii)] For $\gam\in[0,1]$, \eqref{eq:HJB} is ``\emph{well-posed}'' in the sense that it has a unique solution for all values of $T$.
	\item[(iii)] For $\gam\in(1,\gamo]$, \eqref{eq:HJB} is ``\emph{ill-posed}'' (resp. ``\emph{well-posed}'') if $0<\kap<\frac{1}{2}$  (resp. $\kap>\frac{1}{2}$). For $\kap=\frac{1}{2}$, we have $\gamo=1$ and the interval $(1,\gamo]$ is empty. 
\end{enumerate}
These cases are derived from Theorem \ref{thm:HJB} below.
\vspace{5pt}



Achieving infinite expected utility (EU) has been observed in the context of infinite horizon portfolio choice problem (e.g. \cite{Merton1969}), optimal execution (e.g. \cite{bulthuis2016optimal}), and finite horizon optimal trading of assets with mean-reverting returns (e.g. \cite{KimOmberg1996} and \cite{KornKraft2004}). The last two studies, respectively, have coined the terms \mbox{``\emph{nirvana strategies}''} and ``\emph{I-unstable}'' for investment strategies that yield \mbox{infinite}  expected utility over a  finite horizon. 

In reality, however, investors don't achieve infinite EU; otherwise they would drive the prices out of equilibrium. This implies that either: (1) there is no investor with risk tolerance parameter $\gam$ that leads to infinite EU; (2) such an investor exists but the market parameters (that is, $\kap$ and $\rho$) are such that those investors don't have enough time to achieve infinite EU (i.e. $T<T^*(\gam)$ for all agents); or (3) market imperfections (such as transaction cost or parameter ambiguity) prevent investors from achieving infinite EU.

\subsection{Value Function}
Our  solution to  the HJB equation \eqref{eq:HJB} will involve  the solution of the following Riccati equation
\begin{align}
\label{eq:h_ode-1st}
\begin{cases}
	-h'(t) = 2(1-\rho)\gam h^2(t) - 2\frac{\gam\kap}{T+\eps-t}\,h(t) - \frac{(1-\gam)\kap^2}{(1-\rho^2)(T+\eps-t)^2},\\
	h(T)=0.
\end{cases}
\end{align}
There is an  explicit solution to  \eqref{eq:h_ode-1st}, as summarized in Lemma \ref{lem:h}  below.  To prepare for the result, we introduce the following notations. We define
the \emph{``discriminant''} $\Del$ as follows
\begin{align}
	\Del = \Del(\gam) := \frac{\rho-1}{\rho+1}\,\kap^2\,\gam^2 + \left(\frac{2\kap}{1+\rho}-1\right)\kap\gam + \frac{1}{4}.
\end{align}
Furthermore, the \emph{``escape time''} $T^*(\gam)\in(0,+\infty]$ is given by:
\begin{enumerate}
	\item[(i)] If $\gam>\gamo$, then
	\begin{align}
		T^*(\gam) = \eps\left[\exp\left(\frac{\pi}{2\sqrt{-\Del}}-\frac{1}{\sqrt{-\Del}}\arctan\left(\frac{1-2\gam\kap}{2\sqrt{-\Del}}\right)\right) - 1\right].
	\end{align}
	
	\item[(ii)] If $\gam=\gamo$ and $0<\kap<\frac{1}{2}$, then
	\begin{align}
		T^*(\gam) =\eps\left[\exp\left(\frac{2}{1-2\gam\kap}\right) - 1\right].
	\end{align}
	
	\item[(iii)] If $\gam\in(1, \gamo)$ and $0<\kap<\frac{1}{2}$, then
	\begin{align}
		T^*(\gam) =\eps\left[\left(\frac{0.5 - \gam\,\kap + \sqrt{\Del}}{0.5 - \gam\,\kap - \sqrt{\Del}}\right)^{\frac{1}{2\sqrt{\Del}}} - 1\right].
	\end{align}
	
	\item[(iv)] $T^*(\gam)=+\infty$ for all other values of $\gam$.
\end{enumerate}
The escape time plays a critical role in the Riccati equation \eqref{eq:h_ode-1st}.




\begin{lemma}\label{lem:h}
	The Riccati equation \eqref{eq:h_ode-1st} has a solution only if $T<T^*(\gam)$. In particular, if $\gam>\gamo$, then $\Del(\gam)<0$ and the solution    is  
	\begin{align}
		h(t) = \frac{
				\sqrt{-\Del}
				\tan\left[\sqrt{-\Del}\log\left(1+\frac{T-t}{\eps}\right)+\arctan\left(\frac{1-2\gam\kap}{2\sqrt{-\Del}}\right)\right]
				+ \gam\kap - \frac{1}{2}
				}
				{2\gam(1-\rho)(T-t+\eps)}.
	\end{align}
	If $0\le\gam\le\gamo$, then $\Del(\gam) \ge0$ and the solution is	
	\begin{align}
		h(t) =
		\begin{cases}\displaystyle
			-\frac{\kap^2}{1-\rho^2}\left(\frac{1}{\eps} - \frac{1}{T-t+\eps}\right);
			&\quad \text{if }\gam=0\vspace{1em},
			\\
			\displaystyle
			\left(\frac{\log\left(1+\frac{T-t}{\eps}\right)}{1-(\frac{1}{2} - \gam\kap)\log\left(1+\frac{T-t}{\eps}\right)}\right)\\
			\qquad\qquad\qquad\times\left(\frac{(\frac{1}{2} - \gam\kap)^2}{2\gam(1-\rho)(T-t+\eps)}\right);
			&\quad \text{if }\gam=\gamo\vspace{1em},
			\\
			\left(\frac{\left(1+\frac{T-t}{\eps}\right)^{2\sqrt{\Del}} -1}
			{\left(\sqrt{\Del}-\frac{1}{2} + \gam\kap\right)\left(1+\frac{T-t}{\eps}\right)^{2\sqrt{\Del}}
			+\sqrt{\Del}+\frac{1}{2} - \gam\kap}\right)\\
			\qquad\qquad\qquad\qquad\qquad\times
			\frac{\kap^2(\gam-1)}{(1-\rho^2)(T-t+\eps)};
			&\quad \text{if }0<\gam<\gamo.
		\end{cases}
	\end{align}
	If $T^*(\gam)<+\infty$, then $\lim_{T\to T^*(\gam)}\,|h(0;T)| = +\infty$.  \vspace{1ex}
\end{lemma}

The proof is by direct substitution, and thus  omitted.

\begin{remark}
	Note that as $\eps\to 0^+$, we have $T^*(\gam)\to 0^+$ for cases (i)-(iii). This is consistent with the scenario where there is arbitrage in the limit $\eps\to 0^+$. 
\end{remark}

Now, we present the value function in explicit  form.

\begin{theorem}\label{thm:HJB}
	Assume   $T<T^*(\gam)$. The solution of the HJB equation \eqref{eq:HJB} is
	\begin{align}\label{eq:vp}
		&\vp(t,x,z) = U(x)\,\exp\left(f(t)+g(t)\,z+\frac{1}{2}h(t)\,z^2\right), 
	\end{align} for $(t,x,z)\in [0,T]\times\Dc\times\Rb$,  where $h(t)$ is the solution of \eqref{eq:h_ode-1st} given by Lemma \ref{lem:h}, and $g(t)$ and $f(t)$ are given by
	\begin{align}
	\begin{cases}                      
		g'(t) +\gam\left(2(1-\rho)h(t) - \frac{\kap}{T+\eps-t}\right)g(t)\\
		\qquad= \frac{\kap(1-\gam)\big(\mu_2-\rho \mu_1\big)}{(1-\rho^2)(T+\eps-t)} - \gam(\mu_1-\mu_2)h(t)\vspace{1em},\\
		g(T)=0\vspace{1em},
	\end{cases}\label{eq:g}
	\end{align}
	and
	\begin{align}
		f(t) = &\int_t^T (1-\rho)\gam g^2(u) + \gam(\mu_1-\mu_2) g(u) + (1-\rho) h(u) \dd u\\
		&- \frac{(1-\gam)(\mu_1^2 + \mu_2^2 - 2\rho\mu_2\mu_1)(T-t)}{2(1-\rho^2)}.
		\label{eq:f}
	\end{align} 
\end{theorem}
\begin{proof}
	See Appendix \ref{app:HJB}.
\end{proof}

Next, we verify that the value function coincides with the solution of the HJB equation \eqref{eq:HJB} from Theorem \ref{thm:HJB}. We also identify the optimal trading strategy.

\begin{theorem}\label{thm:VF}
Assume $T<T^*(\gam)$.  The value function $V$ in \eqref{eq:VF}  is equal to the function  $\vp$ given in Theorem \ref{thm:HJB}. Furthermore, the optimal trading strategy, denoted by $\piv^*(t,X^*_t,Z_t)$, is in   feedback form, where 
\begin{align}
	\piv^*(t,x,z) &= 
	\begin{pmatrix}
		\pi^*_1(t,x,z)\\
		\pi^*_2(t,x,z)
	\end{pmatrix}\\\label{eq:pis}
	&=\del(x) \left[
		\begin{pmatrix}
			\frac{\mu_1-\rho\mu_2}{1-\rho^2} + g(t)\\
			\frac{\mu_2-\rho\mu_1}{1-\rho^2} - g(t)
		\end{pmatrix}
		+
		z\,
		\begin{pmatrix}
			h(t) - \frac{\kap\,\rho}{1-\rho^2}\\
			-h(t) + \frac{\kap}{1-\rho^2}
		\end{pmatrix}
	\right],
\end{align}
for $(t,x,z) \in [0,T]\times\Dc\times\Rb.$ Here, $\del(x)$ is the risk tolerance function defined in \eqref{eq:HARA}. 
\end{theorem}

\begin{proof}
	See Appendix \ref{app:VF}.
\end{proof}

\begin{figure}[tb]

\centerline{
\adjustbox{trim={0.0\width} {0.0\height} {0.07\width} {0.2\height},clip}
{\includegraphics[scale=0.25, page=1]{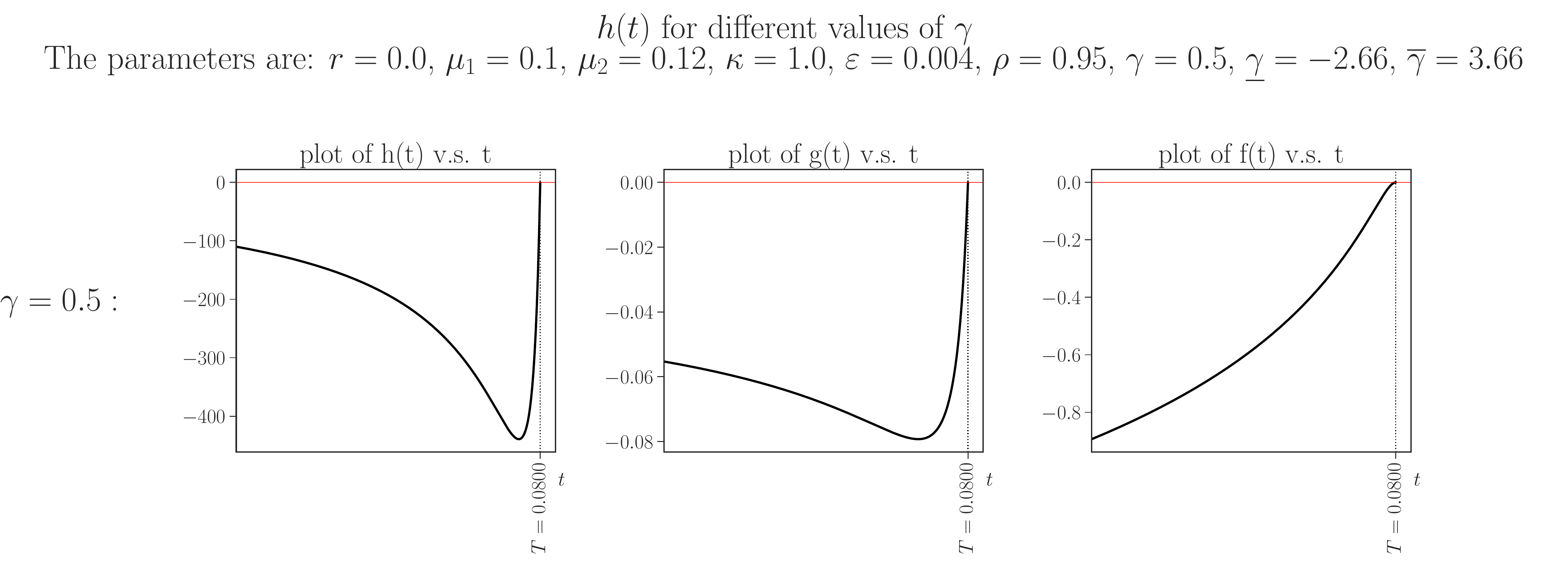}}
}

\centerline{
\adjustbox{trim={0.0\width} {0.0\height} {0.07\width} {0.2\height},clip}
{\includegraphics[scale=0.25, page=1]{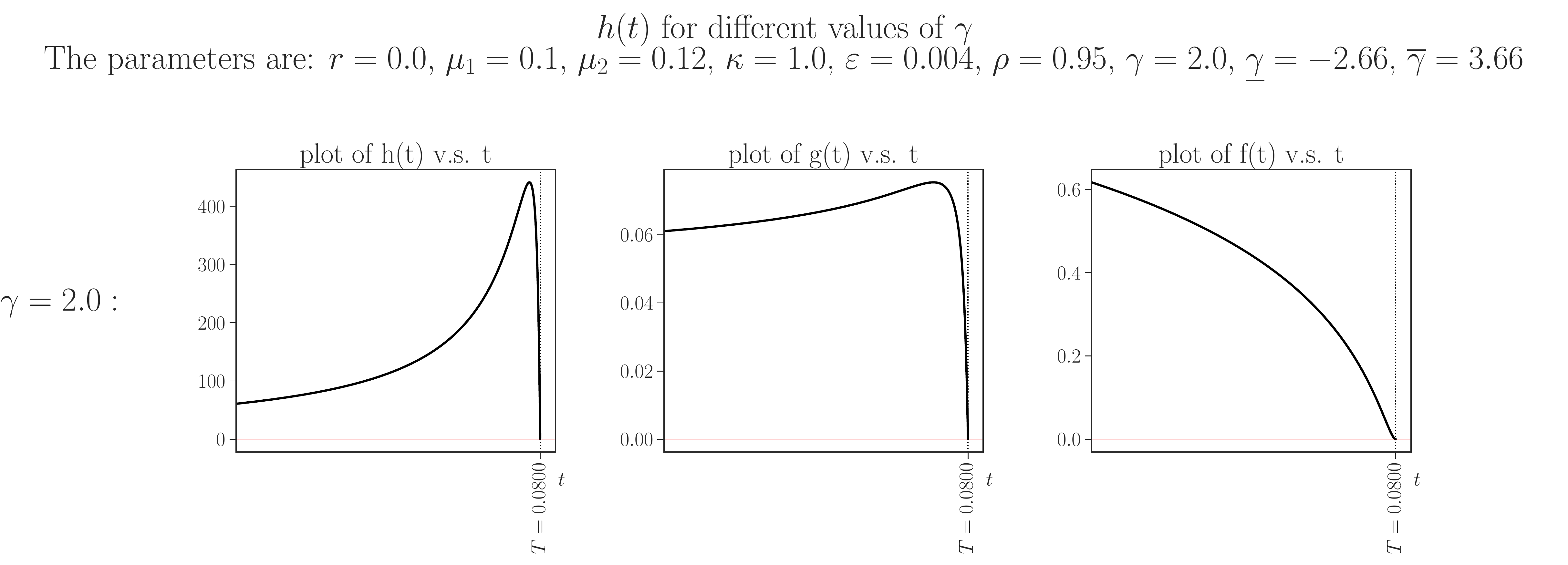}}
}

\centerline{
\adjustbox{trim={0.0\width} {0.0\height} {0.07\width} {0.2\height},clip}
{\includegraphics[scale=0.25, page=1]{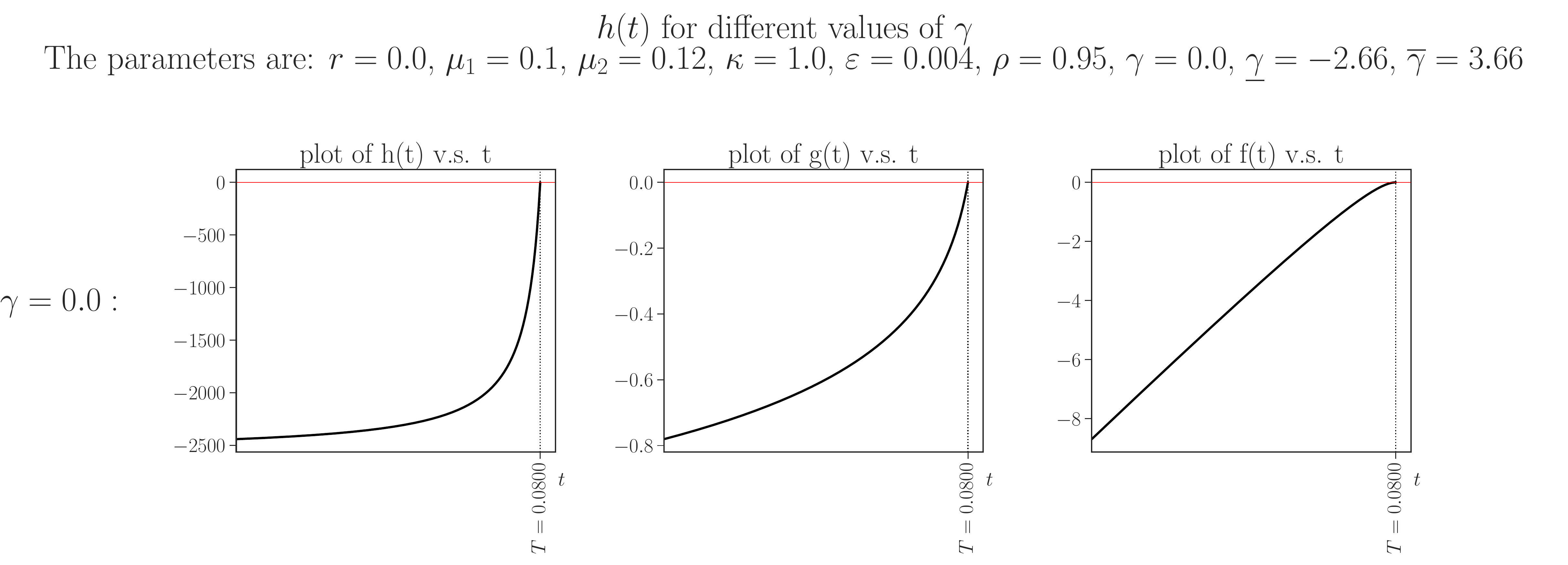}}
}
\caption{Plots of $h(t)$, $g(t)$, and $f(t)$ for various values of $\gam$ in the well-posed range. The model parameters are $\mu_1=0.1$, $\mu_2=0.12$, $\kap=1$, $\eps=0.004$, $\rho=0.95$, and $T=0.08$. These parameter yield the critical value of $\gamo=3.66$.}
\label{fig:wellposed}
\end{figure}

Figures \ref{fig:wellposed} illustrates the behaviors of   functions $f$, $g$, and $h$ under  three well-posed scenarios. We have set $\kap=1$ and $\rho=0.95$  such that the critical risk tolerance value is $\gamo=3.66$. Each row of the figure shows the functions for different value of $\gam$, namely $0$, $0.5$, and $2$. Note that since $\kap=1>0.5$, the values of $\gam$ in the interval $(1,\gamo=3.66)$ are well-posed.   In the exponential case (i.e. $\gam=0$), $h$ monotonically increases to zero. For the power cases (i.e. $\gam\in\{0.5, 2\}$), $h$ is not monotone and   $|h|$ reaches a maximum of high value before reaching zero. For large values of $T-t$ (i.e. the left end on the x-axis), $h(t)$ seems to flatten. In terms of behaviors, $g$ appears  to have similar properties to  $h$, although the values and variations of $g$ is significantly less. In all three scenarios, $f$ moves monotonically to zero over time.

\begin{figure}[ht]
\centerline{
\adjustbox{trim={0.06\width} {0.0\height} {0.06\width} {0.2\height},clip}
{\includegraphics[scale=0.28, page=1]{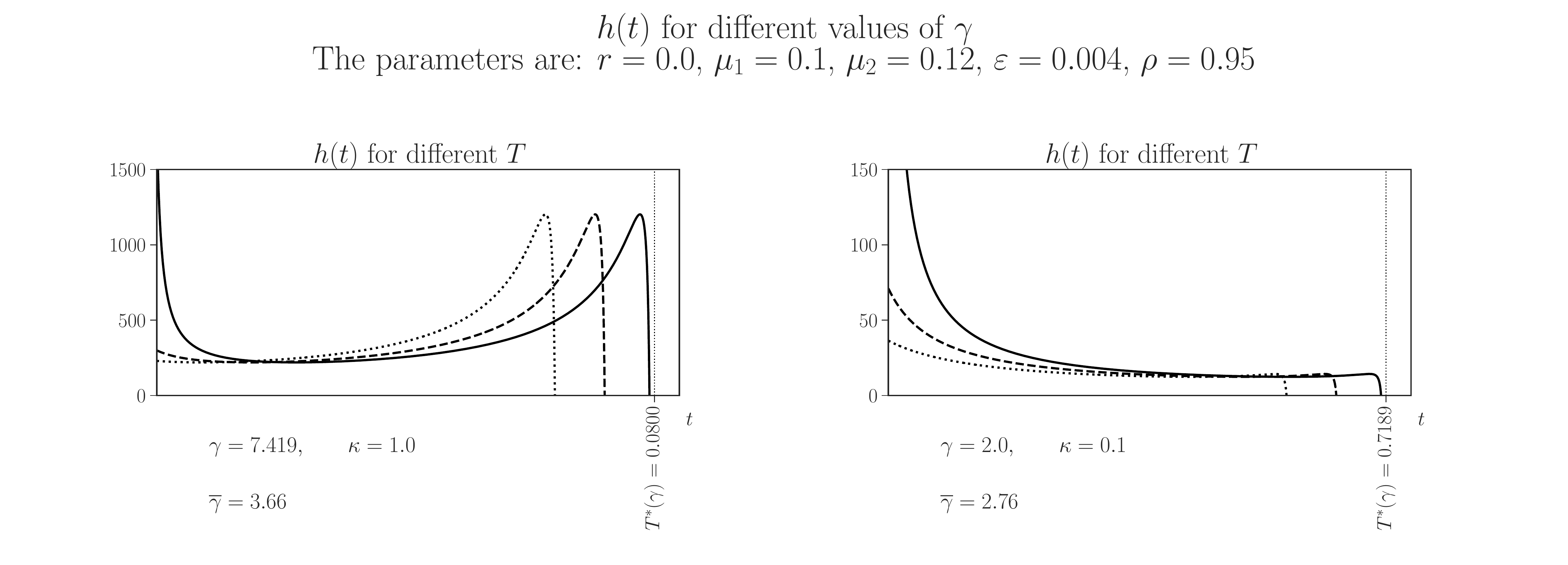}}
}
\caption{Plots of $h(t)$ for two ill-posed scenarios. In each scenario, three values of $T$ are, respectively,  $80\%$ (dotted curve), $95\%$ (dashed curve), and $99\%$ (solid curve) of the escape time $T^*(\gamma)$. The chosen values of $\gam$, $\kap$, and $\gamo$ are stated under each plot. The remaining parameters have the same values as in Figure \ref{fig:wellposed}.}
\label{fig:illposed}
\end{figure}

In Figure \ref{fig:illposed}, we consider  two ill-posed scenarios. Specifically, we plot $h(t)$ over time for  $\gam>\gamo$ (the left panel)
and $1<\gam<\gamo$ and $0<\kap<0.5$ (the right panel). In each plot, $h(t)$ is plotted for three choices of $T$, which are, respectively, $80\%$ (dotted curve), $95\%$ (dashed curve), and $99\%$ (solid curve) of the escape time $T^*(\gamma)$. As expected for these ill-posed cases, we see that  $\lim_{T\to T^*(\gam)^-} |h(0)|=+\infty$, leading the expected utility to reach $+\infty$ over a finite investment horizon.

From Figures \ref{fig:wellposed} and \ref{fig:illposed}, we observe  that $h$ is positive for $\gam>1$, and negative for $0\le\gam<1$.\footnote{Recall that $\gam=1$ corresponds to the logarithmic utility which we have excluded from our analysis.} This motivates us to end the section by  showing it analytically. The proof is provided in Appendix. 

\begin{lemma}\label{lem:sign}
	Consider $h(t)$, the solution of \eqref{eq:h_ode-1st}, given by Lemma \ref{lem:h}. If $\gam<1$ (resp. $\gam>1$), then $h(t)<0$ (resp. $h(t)>0$), for all $t\in[0,T)$. 
\end{lemma}

%
%
\section{Optimal basis trading strategy}\label{sec:Opt}
Let us   examine the optimal trading strategy $\piv^*$ in  \eqref{eq:pis}.  Recall from Remark \ref{rem:VolNorm} that $\piv^*(t,X^*_t,Z_t)$ is the volatility adjusted position. The actual position value is given by $\pivt_t = (\frac{\pi^*_1(t,X^*_t,Z_t)}{\sig_{1,t}}, \frac{\pi^*_2(t,X^*_t,Z_t)}{\sig_{2,t}})$, where $(\sig_{1,t})$ and $(\sig_{2,t})$ are the volatility of the spot and futures prices, respectively. 
 
Our first observation is that $\piv^*$ is directly  proportional to the risk-tolerance function $\del(x)$, as has been observed in other studies such as \cite{karatzas1998methods} and \cite{zariphopoulou2001solution}. In other words, the larger the risk-tolerance, the larger the optimal positions in  \eqref{eq:pis}. If we consider the  special case with $\kappa =0$, then it follows
 that  $g\equiv h\equiv 0$, and the optimal strategy $\piv^*$ reduces to 
\begin{align}
	\piv_M := 
	\del(x)
	\begin{pmatrix}
		\frac{\mu_1-\rho\mu_2}{1-\rho^2}\\
		\frac{\mu_2-\rho\mu_1}{1-\rho^2}
	\end{pmatrix}.
\end{align}
This simple strategy does not depend on $z$. This makes sense because,  when $\kappa=0$, $Z_t$ disappears in the SDE for $F_t$ (see \eqref{eq:F}).   One can loosely interpret $\piv_M$ as the Merton strategy. More generally, when $\kappa \neq 0$, the  incorporation of  the scaled Brownian bridge $Z_t$ in $F_t$ makes the optimal strategy $\piv^*$ dependent on  the time varying functions $g(t)$ and $h(t)$ as well as the stochastic factor $Z_t$.

According to \eqref{eq:pis}, for fixed values of $x$ and $t$,  $\pi^*_1(t,x,z)$ and  $\pi^*_2(t,x,z)$ are affine functions of $z$, with the slopes given by $h(t) - \frac{\kap\,\rho}{1-\rho^2}$ and $-h(t) + \frac{\kap}{1-\rho^2}$ respectively. The slope of this function has an interesting interpretation. Note that a positive log basis $Z_t=\log(S_t/F_t)$ means that  the spot is priced higher than the futures. To take advantage of the anticipated convergence, our strategy is to take   a short position in $S$ and a long position in $F$. The optimal strategy implies that for larger basis  $Z_t$,  the position in $S$ becomes more negative  while the position in $F$ becomes more positive. Similarly, for negative values of $Z_t$ (i.e. $F_t >S_t$),   one expects to take a short position in $F$ and a long position in $S$. 

For $\piv^*$ to be a convergence trading strategy, we expect that  $h(t) - \frac{\kap\,\rho}{1-\rho^2}\le 0$ and $-h(t) + \frac{\kap}{1-\rho^2}\ge0$ or, equivalently,
\begin{align}\label{eq:longshort}
	h(t) \le \min\{\frac{\kap\,\rho}{1-\rho^2},  \frac{\kap}{1-\rho^2}\} = \frac{\kap\,\rho}{1-\rho^2}.
\end{align}
Let's assume that $\rho\ge0$, that is, the futures and spot prices are non-negatively correlated. Then, \eqref{eq:longshort} is always satisfied near maturity, since $\lim_{t\to T^-} h(t) = 0$. In other words, the optimal strategy always becomes a convergence trading strategy once we are close enough to maturity. It may also be the case that \eqref{eq:longshort} is satisfied for all values of $t\in[0,T]$. One sufficient condition is $\gam\le1$ (and $\rho\ge0$), since Lemma \ref{lem:sign} yields that $h(t)\le0$ for all $t$. On the other hand, it is also possible for  \eqref{eq:longshort} to be violated for some parameter values or for $t$ far enough from maturity. We show an example of such a case below.

Figure \ref{fig:VF_gamhalf} illustrates the optimal value function and  optimal trading strategies    for the HARA utility   $U(x) = -\frac{1}{x}$ (i.e. $\gam=0.5$). Setting $x=1$, we consider the  contour plot of the value function for different values of $(t,z)$. It is decreasing in $t$ and increasing as $z$ deviates from zero.  This is intuitive since a longer time to maturity or larger basis implies more potential for making profits. The 95\% confidence region of $(Z_t)$, i.e. $\{(t,z): |z-m(t)| \le 1.96 * \sig(t,t)\}$,  shows that $(Z_t)$ has larger variance near the midpoint of the trading horizon.

The top and bottom right plot of Figure \ref{fig:VF_gamhalf} illustrate the optimal positions $\pi^*_1(t,x,z)$ and $\pi^*_2(t,x,z)$ over values of $(t,z)$ in the 95\% confidence region. Note that for fixed $z$ and $x$ and as $t\to T$, position sizes (i.e. $|\pi^*_i(t,x,z)|$, $i=1,2$) increase and then decrease. The reason for this behavior can be seen from the corresponding function $h$ in Figure \ref{fig:wellposed}, where $|h(t)|$ reached a maximum before vanishing at $T$. Furthermore, note that for fixed $x$ and $t$, $\pi^*_1(t,x,z)$ is decreasing in $z$, while $\pi^*_2(t,x,z)$ is increasing in $z$. As mentioned above, this is because $h(t)<0$ and \eqref{eq:longshort} is satisfied.

  Figure \ref{fig:VF_exp} corresponds to the case with exponential utility (i.e. $\gam=0$). Note that the value function and   optimal positions are much less sensitive to $t$ up till near maturity. Furthermore, the positions appear   to  change monotonically as time approaches maturity. This is in contrast to the    power utility case where  position sizes first peak  then decrease towards maturity.

\begin{figure}[tb]
\centerline{
\adjustbox{trim={0.03\width} {0.05\height} {0.09\width} {0.1\height},clip}
{\includegraphics[scale=0.27, page=1]{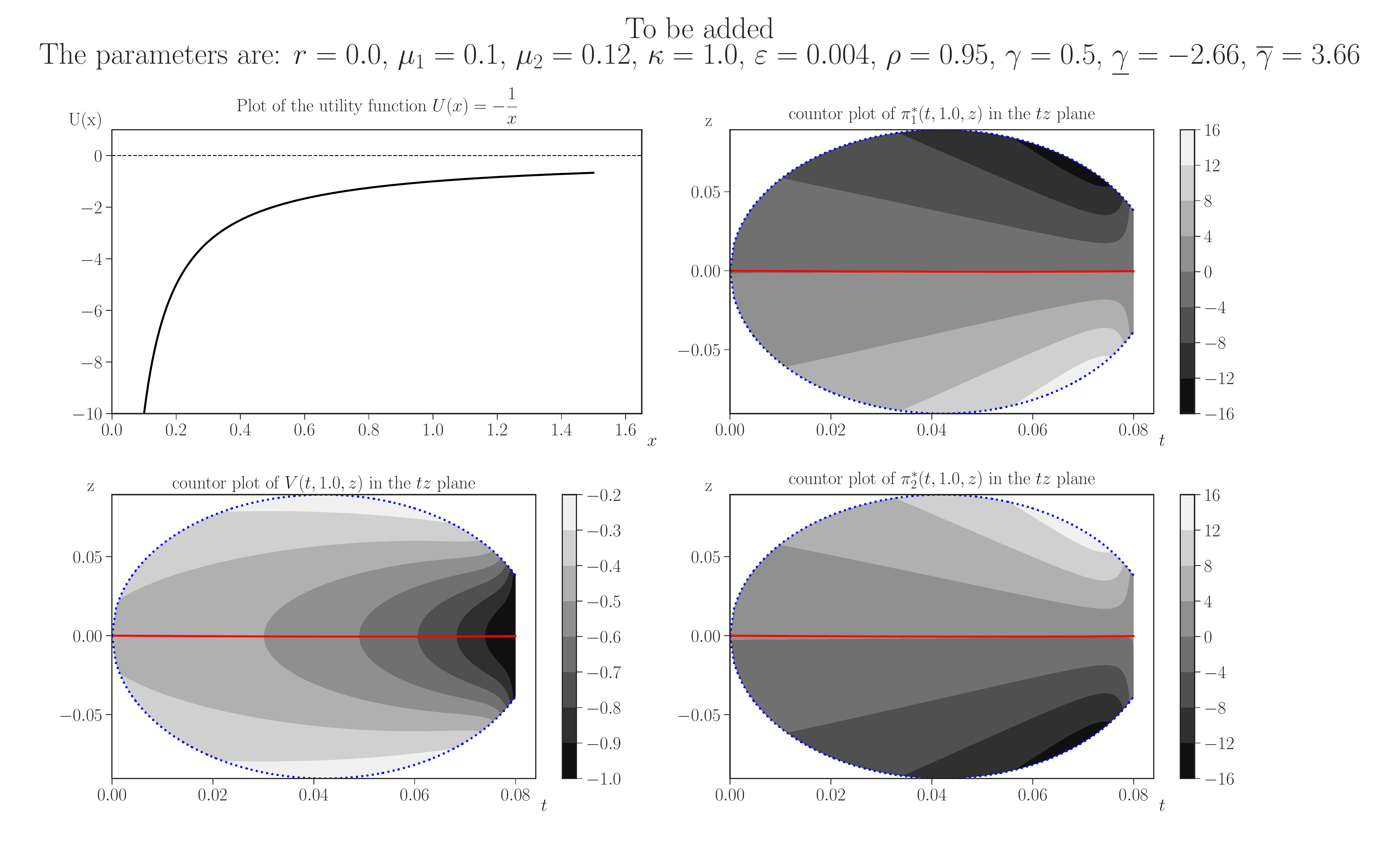}}
}
\caption{The value function and   optimal trading strategy, with  the HARA utility function $U(x) = -\frac{1}{x}$ (i.e. $\gam=0.5$; top left).  Bottom left: contour plot of the value function $V(t,x,z)$ over $(t,z)$ with $x=1.0$, shaded in the 95\% confidence region of $(Z_t)$, i.e. $\{(t,z): |z-m(t)| \le 1.96 * \sig(t,t)\}$ (see Figure \ref{eq:SFZ_sim}). Top and  bottom right panels are the respective contour plots of $\pi^*_1(t,1,z)$ and  $\pi^*_2(t,1,z)$ in the 95\% confidence region of $(Z_t)$.
\label{fig:VF_gamhalf}}
\end{figure}

\begin{figure}[tb]
\centerline{
\adjustbox{trim={0.03\width} {0.05\height} {0.09\width} {0.11\height},clip}
{\includegraphics[scale=0.27, page=1]{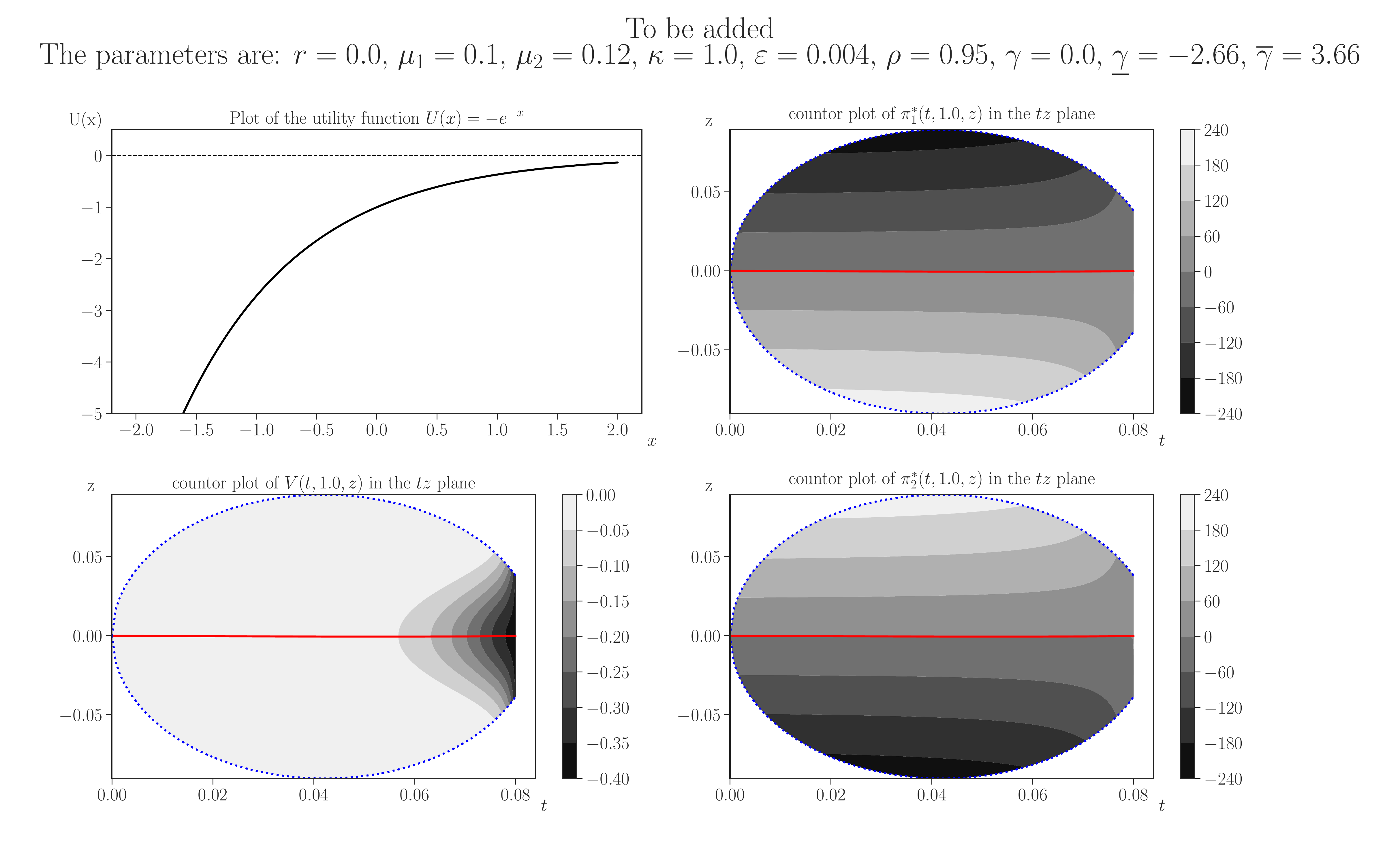}}
}
\caption{The counterpart of Figure \ref{fig:VF_gamhalf} for the exponential utility function $U(x) = -\ee^{-x}$ (i.e. $\gam=0$).
\label{fig:VF_exp}}
\end{figure}

In the previous two examples with $\gam = 0.5$ and $\gam =0$,  the optimal strategies implies convergence trading in that $\pi^*_1$ is  decreasing  and $\pi^*_2$ is  increasing in $z$. This was expected since $\rho\ge0$ and $\gam\le1$ and \eqref{eq:longshort} was satisfied.

Next, we consider a case where \eqref{eq:longshort} does not hold for all values of $t$. With  the utility function  $U(x) = 2\sqrt{x}$ (i.e. $\gam=2$), the chosen parameter values, $\kap=1$ and $\rho=0.95$, imply that  $\frac{\kap\,\rho}{1-\rho^2}\approx 9.74$. Thus, $h(t)$ does not satisfy \eqref{eq:longshort} (see also Figure \ref{fig:wellposed}). From Figure \ref{fig:VF_gam2_zoom},  we observe that at time $t$ away from maturity (say, $t=0.0799$), $\pi^*_1(t,x,z)$ is increasing in $z$ while $\pi^*_2(t,x,z)$ is decreasing. In particular, for large positive values of $z$, meaning that the spot is higher priced than the futures, it is optimal to long   $S$ and short $F$. This is in contrast to a typical convergence trading strategy. Nevertheless, the optimal strategy from our model eventually becomes convergence trading as time approaches maturity. 

To better understand the  optimal trading strategy near maturity, we note that the   mean reversion rate of the Brownian bridge, i.e. $\kap/(T+\eps-t)$, is time varying and becomes much larger near  maturity $T$. Therefore, the force of mean reversion will quickly eliminate any deviation of the stochastic basis from its mean. As a result, convergence trading is optimal near maturity. However, further away from maturity, the mean reversion rate is much smaller, so  any deviation of the stochastic basis from its mean would take longer to get corrected. As it turns out, for a sufficiently risk seeking investor (i.e. $\gam>1$), it is optimal to bet on the deviation from the mean \emph{not} to get immediately corrected. The optimal position for this case is based on the expectation that the basis will diverge (or not immediately converge). Hence, the position is the opposite of a convergence trading.

%
%


\begin{figure}[tb]
\centerline{
\adjustbox{trim={0.06\width} {0.0\height} {0.11\width} {0.3\height},clip}
{\includegraphics[scale=0.23, page=1]{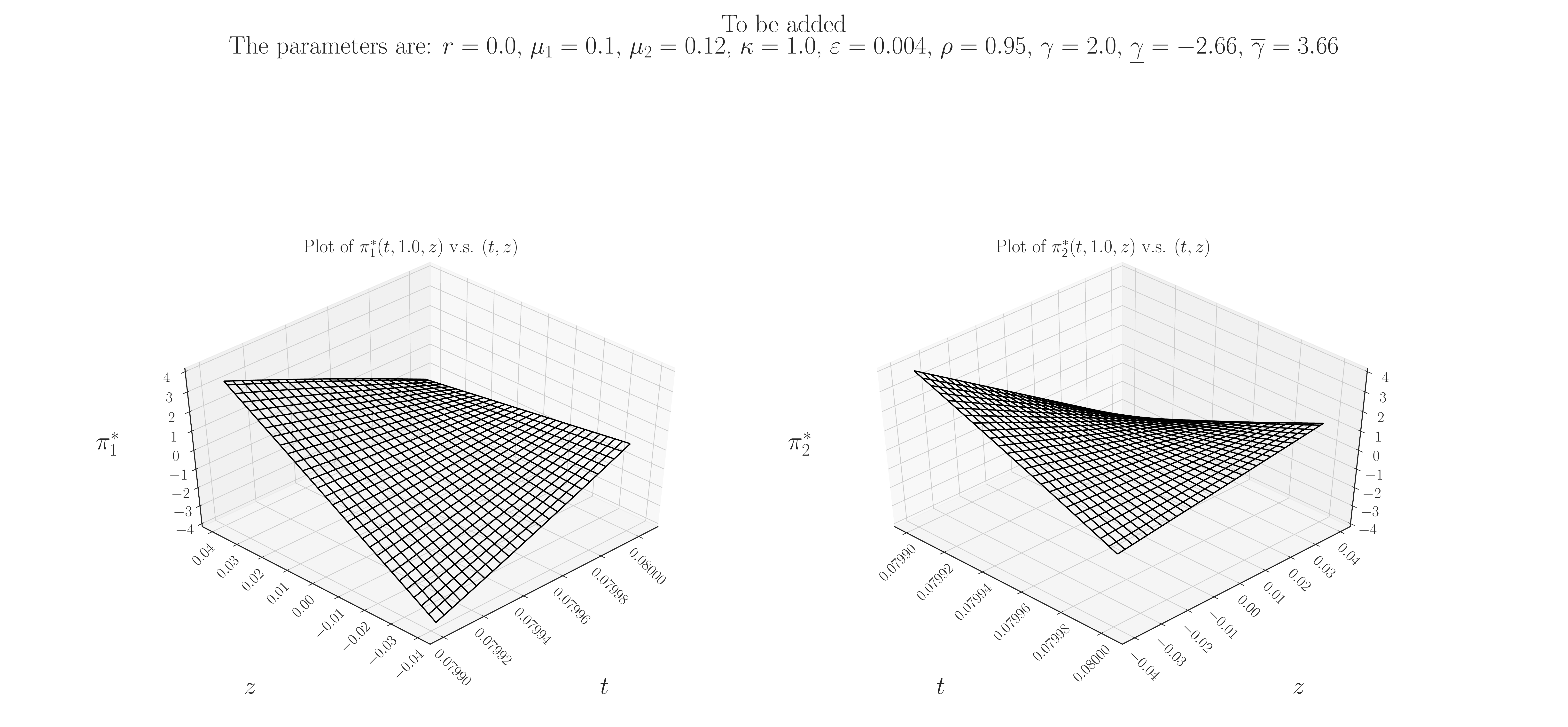}}
}
\caption{Surface plots of the optimal trading strategies, $\pi^*_1(t,1,z)$ (left) and  $\pi^*_2(t,1,z)$ (right), for $0.0799\le t\le 0.08=T$ and $|z-m(t)| \le 1.96 * \sig(t,t)$ under  the utility  $U(x) = 2\sqrt{x}$ ($\gam=2$). At time $0.08$ (maturity), the   strategy trading on convergence     in the sense that $\pi^*_1$ is decreasing in $z$, while $\pi^*_2$ is increasing in $z$. However, at $t=0.0799$, the   strategy is not convergence trading, as $\pi^*_1$ is increasing in $z$ while $\pi^*_2$ is decreasing in $z$.
\label{fig:VF_gam2_zoom}}
\end{figure}

\begin{figure}[tb]
\centerline{
\adjustbox{trim={0.05\width} {0.05\height} {0.09\width} {0.11\height},clip}
{\includegraphics[scale=0.25, page=1]{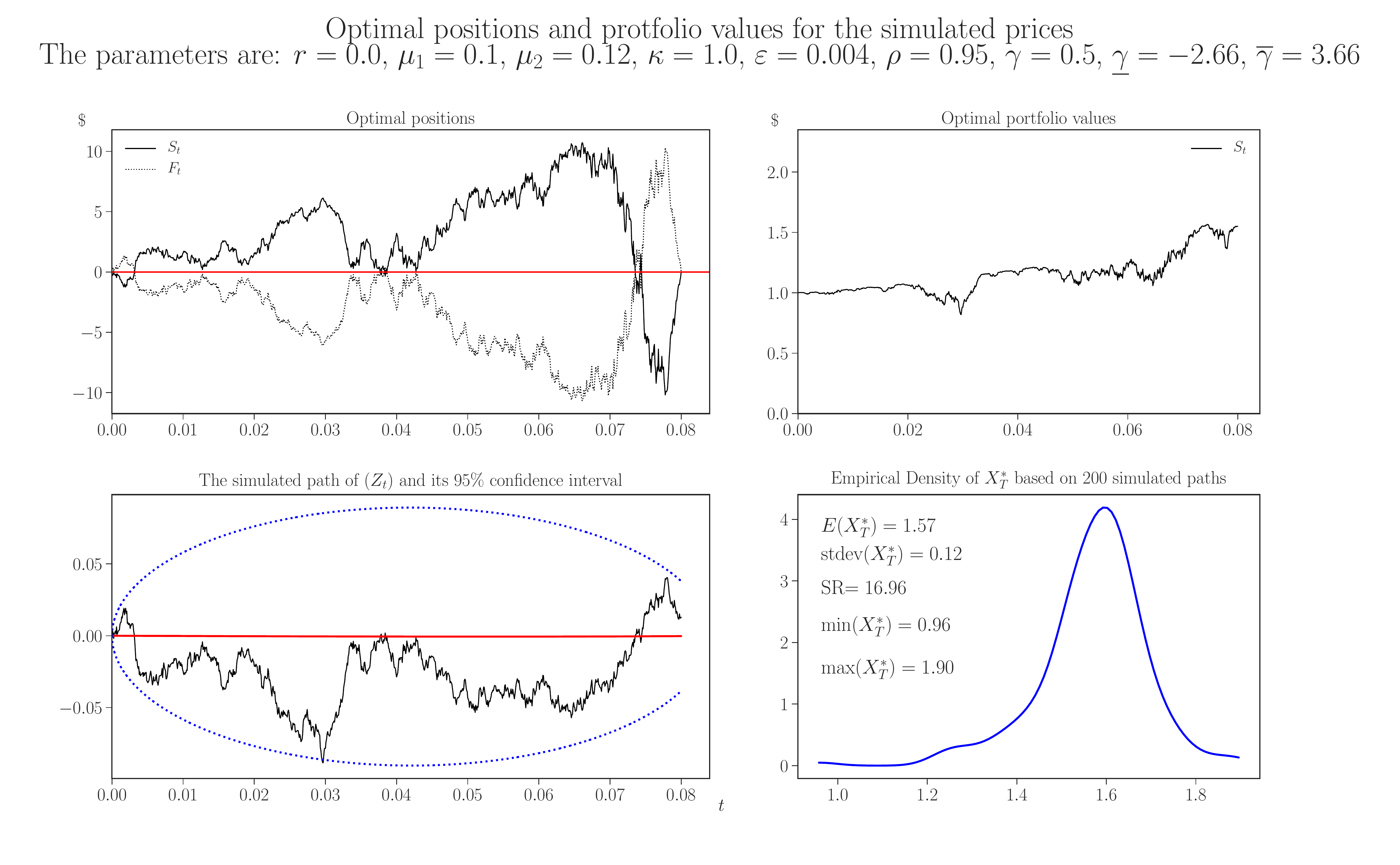}}
}
\caption{Optimal positions in $S$ and $F$ (top left) for one simulated path of $Z$ (bottom left), with $\gam=0.5$ and $x^*=0$.    Top right: the corresponding path of the portfolio value.  Bottom right: the empirical density of the optimal terminal wealth $X^*_T$ bases on 200 simulated price paths, all of which started with the initial wealth $X^*_0=1$. Values of other parameters  are taken from  Figure \ref{eq:SFZ_sim}.
\label{eq:pispf_sim}}
\end{figure}

Figure \ref{eq:pispf_sim} illustrates the optimal positions in $S$ and $F$, along with the    corresponding portfolio value, based on the simulated paths and parameters from Figure \ref{eq:SFZ_sim}. In this   well-posed scenario  with $U(x) = -\frac{1}{x}$ ($\gam=0.5$),   the optimal positions take opposite signs and tend to  move in opposite directions. Moreover, whenever the basis    $Z$ is negative, the position in $S$ is positive and  position in $F$ is negative. The positions are reversed when $Z$ is positive. Furthermore,  the optimal positions fluctuate increasingly more rapidly and more sensitive to the basis near  maturity. 


In Figure \ref{eq:pispf_sim} we also see  the optimal portfolio value over time. Note that portfolio experience significant drawdowns as $Z$ diverges from equilibrium. This is a common characteristic of convergence trading strategies. The bottom right plot of Figure \ref{eq:pispf_sim} shows the empirical density of the optimal terminal wealth $X^*_T$, based on 200 simulated paths, all of which started with the initial wealth $X^*_0=1$. The estimated expected value of terminal wealth is $\E(X^*_T) = 1.57$, with an standard deviation of $0.12$. Since the trading horizon is 20/250 = 0.08 year (i.e. one month), the annualized net expected return is $0.57/0.08 \approx 7.1$ and the annualized volatility is $0.12/\sqrt{0.08}\approx0.42$. This leads to a Sharpe ratio of $7.1/0.42\approx 17$ for this simulated example.


\begin{figure}[h]
\centerline{
\adjustbox{trim={0.04\width} {0.2\height} {0.05\width} {0.25\height},clip}
{\includegraphics[scale=0.375, page=1]{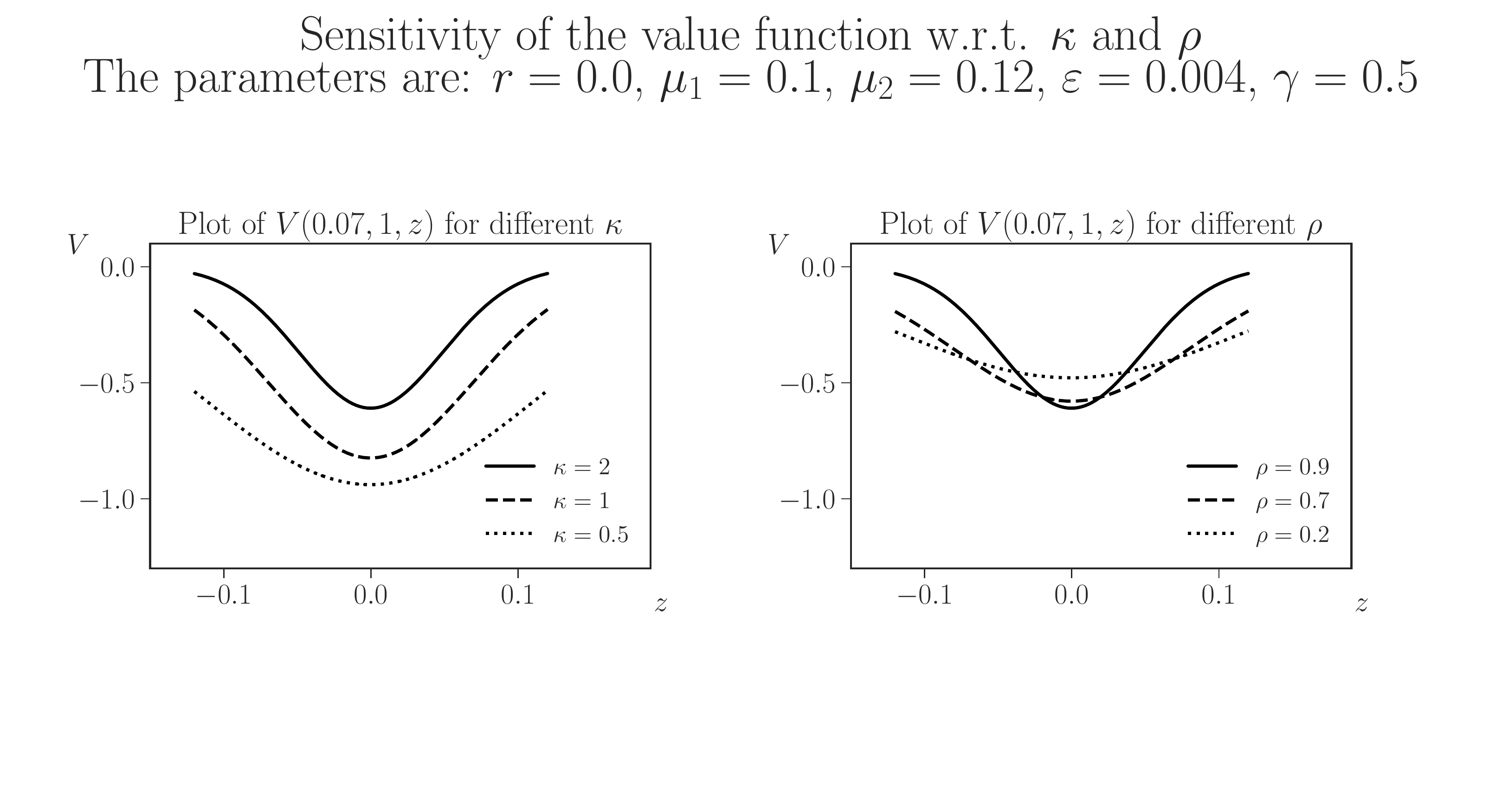}}
}
\caption{Sensitivity of the value function with respect to   $\kap$ and $\rho$.  With  $t=0.07$ and $x=1.0$ fixed, we show  $V(0.07,1,z)$ as a function of  $z$ (left) for three values of $\kap$   and a fixed value of $\rho=0.9$, and (right) for three values of $\rho$ with   $\kap=2$. Other than $\kap$ and $\rho$,  parameters values are as in Figure \ref{fig:VF_gamhalf}.
\label{fig:VF_kaprho}}
\end{figure}

Turning to  the value function, we show in Figure \ref{fig:VF_kaprho} its sensitivity with respect to $\kappa$ and $\rho$. For simplicity, we keep $t=0.07$ and $x=1.0$ fixed and show $V(0.07,1,z)$ as a function of $z$. All other parameters (including the utility function) are as in Figure \ref{fig:VF_gamhalf}. The value function tends to reach its minimum value at $z=0$, and a higher $\kappa$ means a higher value function in the figure. This is intuitive since a higher speed of mean reversion indicates higher profitability from trading the basis at any given   level. On correlation increases, the expected utility improves in case of a large mispricing in the assets (i.e. large values of $|z|$) and it decreases when mispricing is small (i.e. for values $|z|$ near zero). This also makes sense. Note that by \eqref{eq:Bb}, the volatility of the basis $(Z_t)$ is $2(1-\rho)$, which is decreasing in $\rho$. Therefore, large correlation indicates less noise in the basis. If there is no mispricing, less noise means less probability of mispricing which lowers the profitability of basis trading. However, if there is a large mispricing, then less noise means that the mispricing will be corrected faster, which increases the profitability and the expected utility.

\begin{figure}[h]
\centerline{
\adjustbox{trim={0.03\width} {0.2\height} {0.05\width} {0.23\height},clip}
{\includegraphics[scale=0.375, page=1]{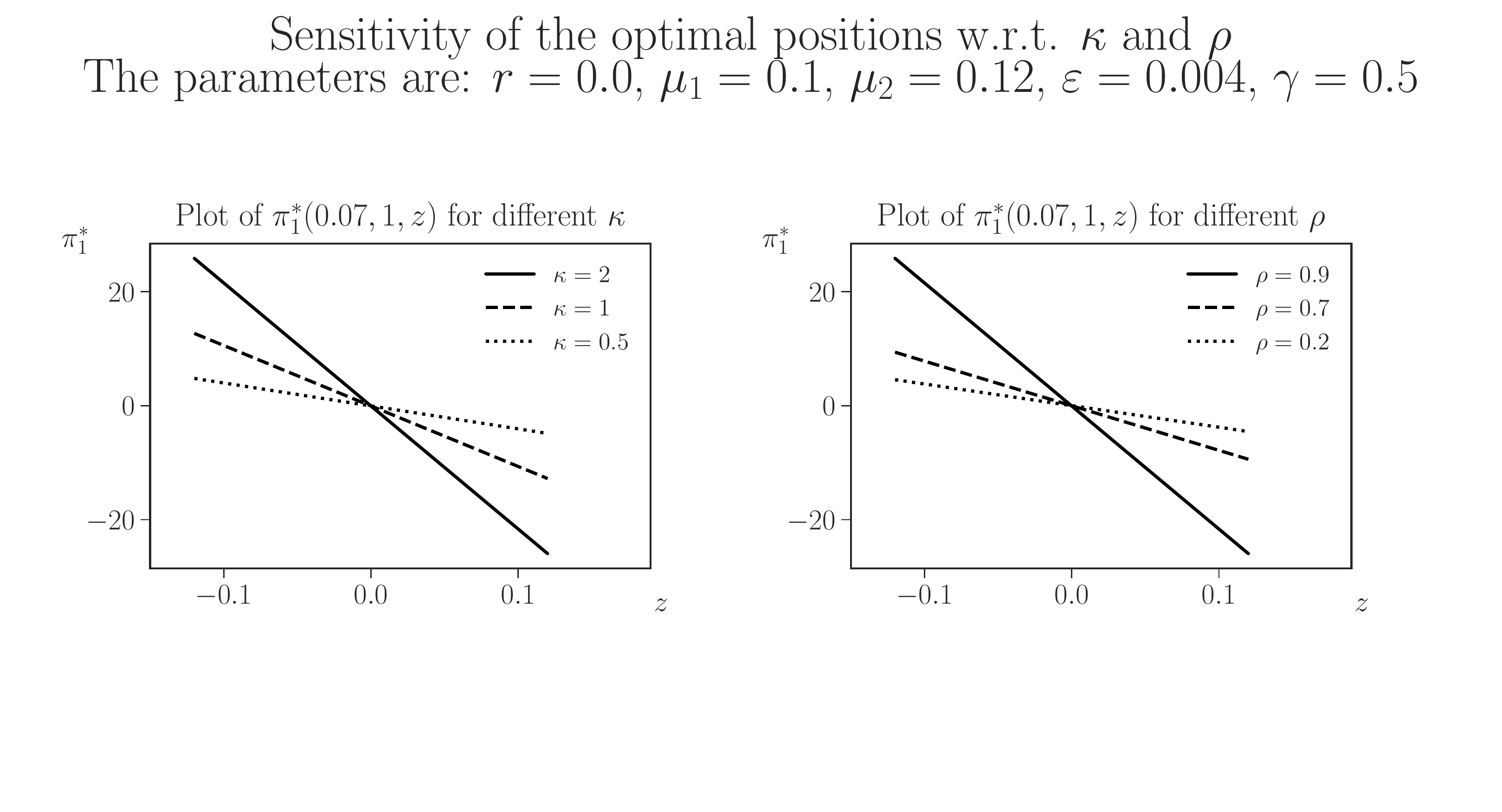}}
}
\caption{Sensitivity of the optimal spot position  with respect to  $\kap$ and $\rho$. Setting $t=0.07$ and $x=1$, we show $\pi^*_1(0.07,1,z)$ as a function of basis $z$ (right) for different  $\kap$ with $\rho=0.9$, and (right)  for different  $\rho$ with $\kap=2$. A higher mean reversion rate or higher correlation increases the long/short position size when $z$ is negative/positive. Other than  $\kap$ or $\rho$,   parameters values are the same as in Figure \ref{fig:VF_gamhalf}. 
\label{fig:piStar1_kaprho}}
\end{figure}

Figure \ref{fig:piStar1_kaprho} shows the sensitivity of the optimal spot position  with respect to the  parameters $\kap$ and $\rho$. For different values of  $\kap$ and $\rho$, we see that  $\pi_1^*$ is a decreasing function of $z$.   A higher   mean-reversion rate $\kap$ or higher correlation $\rho$ makes $\pi_1^*$ more downward sloping in $z$, which means that the investor will take a larger long (resp. short) position  when $z$ is negative (resp. positive). The financial intuition is as follows. A higher $\kap$ means that the basis tends to converge to zero faster. This represents a more profitable trade,  leading the investor to take a larger position in the basis. When the basis is negative, the spot price is below the futures price, so the optimal strategy is to long the spot. The opposite holds when the basis is positive. Increasing $\rho$ decreases the random fluctuations   in the basis,  which also implies that the basis will show  a stronger    tendency to converge. Hence,  it is optimal to take a larger position.

\section{Concluding Remarks}\label{sect-conclude}

In summary, we have analyzed dynamic trading problem in which a risk-averse investor trades the stochastic  basis to maximize expected utility. We describe the non-convergent basis process by a stopped scaled Brownian bridge.  This leads us to solve analytically and numerically the associated HJB equation and illustrate the optimal trading strategies.

There are a number of directions for future research. In addition to basis trading, one can analyze other futures trading strategies, including  futures rolling, and  futures portfolios. Futures are also commonly used in many exchange-traded funds  (ETFs) for tracking the spot price (see e.g. \cite{LeungWard}). It would also be interesting to consider variations of our model. For example, in \eqref{eq:S} and \eqref{eq:F}, one could interpret that the spot price is instantaneously leading  the futures price in  that the futures does not have any feedback on the spot. This is a question of price discovery, and we refer to the empirical studies, \cite{chan1992further}, \cite{kawaller1987temporal}, and \cite{stoll1990dynamics} for more background. It is possible to adapt our solution method by interchanging the leading roles of $S$ and $F$, and more generally incorporate more sophisticated lead-lag effects between futures and spot prices into the trading problem.

\appendix 

\section{Proof of Proposition \ref{prop:NA}}\label{app:NA}
This proposition follows from existing results, such as Theorem 2.4 of \cite{CheriditoFilipovicYor2005}. For readers' convenience, we provide our own proof in our notations.

From \eqref{eq:S} and \eqref{eq:F}, we have
\begin{align}
	\dd
	\begin{pmatrix}
		S_t\\
		F_t
	\end{pmatrix} = 
	\begin{pmatrix}
		S_t & 0\\
		0   & F_t
	\end{pmatrix}
	\Sig \left(\lamv(t,Z_t) \dd t + \dd\Wv_t\right);\quad 0\le t\le T.
\end{align}
Therefore, by Girsanov's theorem, $\Qb$ is a risk-neutral measure if $(Y_t)_{0\le t\le T}$ given by \eqref{eq:MPR} is a $\Pb$-martingale. It only remains to show that $(Y_t)_{0\le t\le T}$ is a $\Pb$-martingale. 
	
	Define the processes $(\Wt_{t,1},\Wt_{t,2})_{t\ge0}$ as follows,
\begin{align}
\begin{cases}
		\Wt_{t,1} := \sqrt{\frac{1-\rho}{2}}\, W_{t,1} - \sqrt{\frac{1+\rho}{2}}\,W_{t,2},\\
		\Wt_{t,2} := \sqrt{\frac{1+\rho}{2}}\, W_{t,1} + \sqrt{\frac{1-\rho}{2}}\,W_{t,2},		
\end{cases}
\end{align}
and note that $(\Wt_{t,1},\Wt_{t,2})_{t\ge0}$ is a 2-dimensional standard Brownian motion. By substituting $W_{t,i}$ with $\Wt_{t,i}$, $i\in\{1,2\}$, in \eqref{eq:MPR} and noting \eqref{eq:lamb}, one obtains
\begin{align}\label{eq:MPR2}
	Y_t = 1 &+ \sum_{i=1}^2 \int_0^t Y_u \left(\mut_i + \frac{\kapt_i}{T-u+\eps}\,Z_u \right)\dd\Wt_{u,i}; 0\le t\le T,
\end{align}
where $\mut_i$ and $\kapt_i$, $i\in\{1,2\}$, are some appropriately defined constants.
	
	Next, consider the (deterministic) time change
\begin{align}
	\altt(t) &:= 2(1-\rho)\int_0^t (T-u+\eps)^{-2\kap}\dd u\\
	\label{eq:timechange}
	&=2(1-\rho)
	\begin{cases}
		\frac{1}{2\kap-1}\left[(T-t+\eps)^{1-2\kap} - (T+\eps)^{1-2\kap}\right];&\quad \text{if } \kap\ne 0.5,\\
		\log\left(\frac{T+\eps}{T-t+\eps}\right); &\quad \text{if } \kap = 0.5,
	\end{cases}
\end{align}
	and define $\Tt:=\altt(T)<\infty$. Note that $\altt(t)$ strictly increases from 0 to $\Tt$ as $t$ goes from 0 to $T$. Let us also define $t(\altt): [0,\Tt]\to[0,T]$ as the inverse of $\altt(t)$.

	Using Knight's time-change theorem (see Theorem 3.4.13 on page 179 of \cite{KartzasShreve1991}), the processes $(B_{\altt, 1}, B_{\altt,2})_{0\le \altt \le \Tt}$ defined by
\begin{align}\label{eq:B}
	B_{\altt(t), i} := \sqrt{2(1-\rho)}\int_0^t (T-u+\eps)^{-k} \dd \Wt_{u,i};\quad i\in\{1,2\},
\end{align}
are independent standard Brownian motions. Furthermore, from \eqref{eq:Zsol}, it follows that
\begin{align}\label{eq:Zsol2}
	Z_t = Z_0 \left(1-\frac{t}{T+\eps}\right)^\kap + (\mu_1-\mu_2) A(t)
		+ (T-t+\eps)^\kap B_{\altt(t),1}.
\end{align}

Consider the time-changed process $(\Yt_\altt := Y_{t(\altt)})_{0\le\altt\le \Tt}$ and note that $(Y_t)_{0\le t\le T}$ is a martingale if and only if $(\Yt_\altt)_{0\le\altt\le \Tt}$ is a martingale. By \eqref{eq:MPR2}, \eqref{eq:B}, and \eqref{eq:Zsol2}, $\Yt$ satisfies
\begin{align}
	\Yt_\altt = 1 &+ \sum_{i=1}^2\int_0^{t(\altt)} \Yt_v f_i\big(v,B_{v,1}\big)\dd B_{v,i};\quad 0\le \altt \le \Tt,
\end{align}
where we have defined the functional
\begin{align}\label{eq:fi}
	f_i\big(\altt, x(\cdot)\big) :=	\mut_i + 
		&\frac{\kapt_i}{T-t(\altt)+\eps}\\
		&{}\times
	 	\left[
	 		Z_0 \left(1-\frac{t(\altt)}{T+\eps}\right)^\kap + (\mu_1-\mu_2) A\big(t(\altt)\big) + (T-t(\altt)+\eps)^\kap x(\altt)
	 	\right],
\end{align}
for any continuous function $x:[0,\Tt]\to\Rb$. Note that,
\begin{enumerate}
	\item[(i)] $\fv(\altt,x)=\big(f_1(\altt,x), f_2(\altt,x)\big)$ is progressively measurable in the sense of Definition 3.5.15 on page 199 of \cite{KartzasShreve1991}; and,
	
	\item[(ii)] for any $0\le w \le \Tt$, there exists a constant $K>0$ such that
	\begin{align}\label{eq:Fail}
		\|\fv\big(\altt, x(.)\big)\| \le K (1+ \max_{0\le v\le \altt} |x(v)|);\quad 0\le \altt \le w.
	\end{align}
 This  follows from \eqref{eq:fi}, since $\eps>0$ and $A(t)$ and $t(\altt)$ are continuous (and, thus, bounded) on $[0,T]$ and $[0,\Tt]$, respectively.
\end{enumerate}
With conditions (i) and (ii)  above, we now apply Corollary 3.5.16 on page 200 of \cite{KartzasShreve1991} to conclude that $(\Yt_\altt)_{0\le\altt\le \Tt}$ is a martingale. Hence, $(Y_t)_{0\le t\le T}$ is also a martingale, as we set out to prove.\qed
	

%
%
\section{Proof of Theorem \ref{thm:HJB} }\label{app:HJB}
Applying the operator $\Lc_{\piv}$ in \eqref{eq:Opr}  to the function $\vp(t,x,z)$ in \eqref{eq:HJB}, we get
	\begin{align}
		\Lc_{\piv} \vp(t,x,z) = &\frac{1}{2}\vp_{xx} 
		\left\|
		\Sig^\top\piv
		+ \frac{\vp_x}{v_{xx}}
		\Sig^{-1}
		\begin{pmatrix}
			\mu_1\\
			\mu_2 + \frac{\kap\,z}{T-t+\eps}	
		\end{pmatrix}
		+ \frac{\vp_{xz}}{v_{xx}}\Sig^\top
		\begin{pmatrix}
			1\\
			-1		
		\end{pmatrix}
		\right\|^2\\
		& -\frac{1}{2}\left[a\left(\frac{z}{T-t+\eps}\right)^2 + \frac{2b\,z}{T-t+\eps} + c\right]\frac{\vp_x^2}{\vp_{xx}}
		-(1-\rho)\frac{\vp_{xz}^2}{\vp_{xx}}\\
		&-\left(\mu_1-\mu_2-\frac{\kap\,z}{T-t+\eps}\right)\,\frac{\vp_{x}\vp_{xz}}{\vp_{xx}},
	\end{align}
	where we defined the constants
	\begin{align}\label{eq:abc}
		a:= \frac{\kap^2}{1-\rho^2},\quad
		b:= \frac{\kap\big(\mu_2-\rho \mu_1\big)}{1-\rho^2},\quad
		c:= \frac{\mu_1^2 + \mu_2^2 - 2\rho\mu_2\mu_1}{1-\rho^2},
	\end{align}
	and matrix
	\begin{align}
		\Sig :=
		\begin{pmatrix}
			1 &0\\
			\rho &\sqrt{1-\rho^2}
		\end{pmatrix}.
	\end{align}
	Assume, for now, that $\vp_{xx}(t,x,z)<0$, for $(t,x,z)\in[0,T]\times\Dc\times\Rb$, which will hold once we show that the solution is of the desired form \eqref{eq:vp}. 
	Then, it   follows that, for $(t,x,z)\in[0,T]\times\Dc\times\Rb$,
\begin{align}\label{eq:Opt}
	\piv^*(t,x,z) &:= \underset{\piv\in\Rb^2}{\arg\max}\,\Lc_{\piv} \vp(t,x,z)\\
	&=-\frac{\vp_x}{\vp_{xx}}(\Sig\Sig^{\top})^{-1}
	\begin{pmatrix}
		\mu_1\\
		\mu_2 + \frac{\kap\,z}{T-t+\eps}	
	\end{pmatrix}
	-\frac{\vp_{xz}}{\vp_{xx}}
	\begin{pmatrix}
		1\\
		-1		
	\end{pmatrix},
\end{align}
and the HJB equation \eqref{eq:HJB} becomes to the second-order nonlinear PDE
\begin{align}\label{eq:HJB2}
\begin{split}
	\vp_t(t,x,z) &= 
	  \frac{1}{2}\left[a\left(\frac{z}{T-t+\eps}\right)^2 + \frac{2b\,z}{T-t+\eps} + c\right]
	  \frac{\vp_x^2}{\vp_{xx}}\\*
	&\quad
	+ \left(\mu_1-\mu_2-\frac{\kap\,z}{T-t+\eps}\right)
	  	\left(\frac{\vp_{x}\vp_{xz}}{\vp_{xx}} - \vp_z\right)\\
	&\quad
	+ (1-\rho)\left(\frac{\vp_{xz}^2}{\vp_{xx}}-\vp_{zz}\right),
\end{split}
\end{align}
for $(t,x,z)\in[0,T]\times\Dc\times\Rb$ and with the terminal condition $\vp(T,x,z) = U(x)$. Substituting the ansatz\footnote{This ansatz can be derived through the power transformation introduced by \cite{zariphopoulou2001solution}. See, also, \cite{KimOmberg1996}, among others.}
\begin{align}\label{eq:ansatz}
	\vp(t,x,z) = U(x)\,\exp\left(f(t)+g(t)\,z+\frac{1}{2}h(t)\,z^2\right),
\end{align}
$(t,x,z)\in [0,T]\times\Dc\times\Rb$, in \eqref{eq:HJB2} leads to the following identity involving the unknown functions $f$, $g$, and $h$,
\begin{align}
	& f'(t) + (1-\rho)\gam g^2(t) + \gam(\mu_1-\mu_2) g(t) + (1-\rho) h(t) - \frac{1}{2}(1-\gam) c\\
	& + z \left[g'(t) +\gam\left(2(1-\rho)h(t) - \frac{\kap}{T+\eps-t}\right)g(t)
	+ \gam(\mu_1-\mu_2)h(t) - \frac{(1-\gam)b}{T+\eps-t}\right]\\
	&+ \frac{1}{2}\,z^2\left[
	h'(t) +	2(1-\rho)\gam h^2(t) - 2\frac{\gam\kap}{T+\eps-t}\,h(t) - \frac{(1-\gam)a}{(T+\eps-t)^2}
	\right]=0,
\end{align}
for all $t\in[0,T]$ and $z\in\Rb$. To obtain this equation, we have used the identity
\begin{align}
	\frac{(U'(x))^2}{U''(x)} = (1-\gam)\,U(x);\quad x\in\Dc,
\end{align}
which follows directly from \eqref{eq:HARA} and \eqref{eq:HARA2}. Furthermore, substituting \eqref{eq:ansatz} in the terminal condition $\vp(T,x,z) = U(x)$ yields $f(T)=g(T)=h(T)=0$. As a result, $h$, $g$, and $f$ are given by
\begin{align}
\label{eq:hODE}
&\begin{cases}
	-h'(t) = 2(1-\rho)\gam h^2(t) - 2\frac{\gam\kap}{T+\eps-t}\,h(t) - \frac{(1-\gam)a}{(T+\eps-t)^2};
	&\quad t\in[0,T],\\
	h(T)=0,
\end{cases}\\
\label{eq:gODE}                      
&\begin{cases}                        
	g'(t) +\gam\left(2(1-\rho)h(t) - \frac{\kap}{T+\eps-t}\right)g(t)\\
	\qquad ~+ \gam(\mu_1-\mu_2)h(t) - \frac{(1-\gam)b}{T+\eps-t}=0; &\quad t\in[0,T],\\
	g(T)=0,
\end{cases}	
\intertext{and}
\label{eq:fODE}
&\begin{cases}
	-f'(t) = (1-\rho)\gam g^2(t) + \gam(\mu_1-\mu_2) g(t)\\
	\qquad\qquad + (1-\rho) h(t) - \frac{1}{2}(1-\gam) c; &\quad t\in[0,T],\\
	f(T)=0.
\end{cases}
\end{align}

Note that \eqref{eq:hODE} is the Riccati equation \eqref{eq:h_ode-1st}, which we have solved in Lemma \ref{lem:h}. Now that we obtained $h$, we can find $g$ using \eqref{eq:gODE}, and then find $f$ using \eqref{eq:fODE}. The latter equations yield \eqref{eq:g} and \eqref{eq:f} respectively. 

%
%

\section{Proof of Theorem \ref{thm:VF}}\label{app:VF}


Let $\vp$ be the solution of \eqref{eq:HJB}, given by Theorem \ref{thm:HJB}. We prove the following two assertions:
\begin{enumerate}
	\item[(a)] For all $(t,x,z)\in[0,T]\times\Dc\times\Rb$, we have 
	\begin{align}
		v(t,x,z) \ge \Eb_{t,x,z}\left(U(X^\pi_T)\right), \quad \text{for all } \piv\in\Ac,
	\end{align}
	where $X^\pi_T$ is the terminal wealth generated by an admissible strategy  $\piv\in\Ac$.
	\item[(b)] There exists an admissible strategy $\piv^*\in\Ac_i$ such that the corresponding wealth process $(X^*_t)_{t\in[0,T]}$ satisfies
	\begin{align}\label{eq:Optimality}
		v(t,x,z) = \Eb_{t,x,z}\left(U(X^*_T)\right);\quad (t,x,z)\in [0,T]\times\Dc\times \Rb.
	\end{align}
\end{enumerate}
As a consequence, (a) implies  $v\ge V$ and (b) implies $v\le V$. Hence, if these assertions hold, we obtain $V\equiv v$ as desired.

We consider two cases, namely, $0\le\gam<1$ and $\gam>1$.\vspace{1ex}

\noindent\textbf{Case 1 ($0\le\gam<1$):} In this case, by Lemma \ref{lem:sign}, the function $h$ is negative. Therefore, $\vp$ which is of the form
\begin{align}\label{eq:HJB-Form}
	\vp(t,x,z) = U(x)\,\exp\left(f(t)+g(t)\,z+\frac{1}{2}h(t)\,z^2\right),
\end{align}
is bounded in $z$. Since $v$ has polynomial growth in $x$, standard verification results such as Theorem 3.8.1 on page 135 of \cite{FlemingSoner2006} yield assertions (a) and (b) above. Furthermore, the  optimal control is given by \eqref{eq:Opt} which, in turn, yields \eqref{eq:pis}.\vspace{1ex}
		
\noindent\textbf{Case 2 ($\gam>1$):} For these values of $\gam$, Lemma \ref{lem:sign} shows that $h$ is positive. Thus, $v$, the solution of the HJB equation, has exponential growth in $z$. For this case, we provide our own verification result by checking assertions (a) and (b) above.

To show (a), for all $n>0$, define the stopping times
	\begin{equation}
		\tau_n := \inf\Big\{t\in(0,T]: \max\{\int_0^t \|\piv_u\|^2 du, \vert X^\pi_t \vert, \vert Z_t \vert\} > n \Big\}\bigg\},
	\end{equation}
	and note that $\lim_{n\to+\infty}\tau_n=T$, $\Pb$-almost surely. By It\^o's formula, we have
	\begin{align}
	\begin{split}
		v(t,X_{\tau_n},Z_{\tau_n}) ={} &v(t,x,z)\\
		&+ \int_t^{\tau_n} \Big[v_t(u,X_u,Z_u) +
		\left(\mu_1-\mu_2-\frac{\kap\,Z_u}{T-u+\eps}\right)\, \vp_z(u,X_u,Z_u)\\
		&\hspace{4em}
		+ (1-\rho)\vp_{zz}(u,X_u,Z_u) + \Lc_{\piv_u}v(u,X_u,Z_u)
		\Big]\dd u\\
		& + \int_t^{\tau_n} \left[v_x(u,X_u,Z_u)\piv_u 
		+ v_z(u,X_u,Z_u)
		\begin{pmatrix}
			1\\
			-1
		\end{pmatrix}
		\right]^\top
 		\Sig\,
		\dd \Wv_u,
	\end{split}\label{eq:Ito}
	\end{align}
	where $\Sig$ is given by \eqref{eq:abc}. The first integral on the right side is non-positive because $v$ solves \eqref{eq:HJB}. Furthermore, by the definition of $\tau_n$, the integrands of the second integral is uniformly bounded, therefore, taking the conditional expectation on both sides of (\ref{eq:Ito}) yields
\begin{equation}\label{eq:ExpectedIto}
	\Eb_{t,x,z} v(\tau_n,X_{\tau_n},Z_{\tau_n}) \le v(t,x,z).
\end{equation}
Note that for $\gam>1$, \eqref{eq:HARA2} yields that $U>0$. Thus,
\begin{align}
	v(t, X_t,Z_t) = U(X_t)\,\exp\left(f(t)+g(t)\,Z_t+\frac{1}{2}h(t)\,Z_t^2\right)>0.
\end{align}
Finally, by letting $n\to+\infty$, Fatou's lemma yields that $\Eb_{t,x,z} U(X_T) \le v(t,x,z)$.

To show (b), define $\piv^*$ by \eqref{eq:pis}, namely,
\begin{align}\label{eq:optimalPort}
	\piv^*(t,x,z) = (x-x^*) \alv^*(t,z);\quad (t,x,z) \in [0,T]\times\Dc\times\Rb,
\end{align}
where, for ease of notation, we have defined
\begin{align}
	\alv^*(t,z) :=
	\gam\,\left[
		\begin{pmatrix}
			\frac{\mu_1-\rho\mu_2}{1-\rho^2} + g(t)\\
			\frac{\mu_2-\rho\mu_1}{1-\rho^2} - g(t)
		\end{pmatrix}
		+
		z\,
		\begin{pmatrix}
			h(t) - \frac{\kap\,\rho}{1-\rho^2}\\
			-h(t) + \frac{\kap}{1-\rho^2}
		\end{pmatrix}
	\right].
\end{align}
We need to show that $\big(\piv^*(t,X^*_t,Z_t)\big)_{t\in[0,T]}\in\Ac$ and that \eqref{eq:Optimality} is satisfied.

	To show admissibility of $\big(\piv^*(t,X^*_t,Z_t)\big)_{t\in[0,T]}$, we proceed as follows. By \eqref{eq:Budget} and \eqref{eq:optimalPort}, the wealth process $(X^*_t)_{t\in[0,T]}$ corresponding to $(\piv^*_t)$ satisfies
	\begin{align}
		\frac{\dd (X^*_t-x^*)}{X^*_t-x^*} = \alv^*(t,Z_t)^\top
		\begin{pmatrix}
			\mu_1\\
			\mu_2 + \frac{\kap\,Z_t}{T-t+\eps}
		\end{pmatrix}\,\dd t
		+\alv^*(t,Z_t)^\top
 		\Sig\, \dd \Wv_t;\quad t\in[0,T].
	\end{align}
	Since $(X^*_t-x^*)_{t\in[0,T]}$ is a stochastic exponential, conditions (i)-(iii) in Definition \ref{def:Admiss} hold if the following integrability condition is satisfied
	\begin{align}\label{eq:integrability}
		&\int_0^T\left(
		\left\|\alv^*(t,Z_t)^\top
		\begin{pmatrix}
			\mu_1\\
			\mu_2 + \frac{\kap\,Z_t}{T-t+\eps}
		\end{pmatrix}\right\|
		+
		\left\|\alv^*(t,Z_t)^\top
 		\Sig\right\|^2
		\right)\dd t <\infty;\quad \Pb\text{-a.s.}
	\end{align}
	Since $g(t)$ and $h(t)$ are bounded on the interval $[0,T]$, we have
	\begin{align}\label{eq:Estimate}
		\left\|\alv^*(t,Z_t)^\top
		\begin{pmatrix}
			\mu_1\\
			\mu_2 + \frac{\kap\,Z_t}{T-t+\eps}
		\end{pmatrix}\right\|
		+
		\left\|\alv^*(t,Z_t)^\top
 		\Sig\right\|^2
		< C_1 Z_t^2 + C_2 |Z_t| + C_3, 
	\end{align}
	for some positive constants $C_1$, $C_2$, and $C_3$. Thus, \eqref{eq:integrability} holds if $\int_0^T Z^2_t\dd t<\infty$, $\Pb$-almost surely. The latter conditions holds since, by Fubini's theorem,
	\begin{align}\label{eq:Ez2}
		\Eb\int_0^T Z^2_t\dd t = \int_0^T(\sig(t,t)+m(t)^2)\dd t <\infty,
	\end{align}
	where $m(t)$ and $\sig(t,t)$ are bounded functions on $t\in[0,T]$ given by \eqref{eq:mean} and \eqref{eq:Cov} respectively. We have proved that $\big(\piv^*(t,X^*_t,Z_t)\big)_{t\in[0,T]}\in\Ac$.
	
	It only remains to prove \eqref{eq:Optimality}. Since $\piv^*$ satisfies \eqref{eq:Opt}, we have
\begin{align}\label{eq:HJBopt}
\vp_t + \left(\mu_1-\mu_2-\frac{\kap\,z}{T-t+\eps}\right)\, \vp_z 
+ (1-\rho)\vp_{zz} \displaystyle+ \Lc_{\piv^*(t,x,z)} \vp = 0,
\end{align}
for $(t,x,z)\in[0,T]\times\Dc\times\Rb.$ Therefore, by setting $\piv_t=\piv^*_t :=\piv^*(t,X^*_t,Z_t)$ in the argument that yielded \eqref{eq:ExpectedIto}, we obtain
\begin{equation}\label{eq:ExpectedIto2}
	\Eb_{t,x,z} \vp(\tau_n,X^*_{\tau_n},Z_{\tau_n}) = \vp(x,z,t);\quad (t,x,z)\in [0,T]\times\Dc\times \Rb.
\end{equation}
Next, we show that the process $\big(v(t, X^*_t,Z_t)\big)_{t\in[0,T]}$ is a martingale. Once we show this, \eqref{eq:Optimality} follows from \eqref{eq:ExpectedIto2} by taking the limit as $n\to+\infty$.

With a slight abuse of notation, let us define $Y_t:= v(t, X^*_t,Z_t)$, $0\le t\le T$. Set $\piv_t=\piv^*_t$ in \eqref{eq:Ito} and use \eqref{eq:HJBopt} to obtain
\begin{align}\label{eq:Ito2}
	\dd Y_t &=
	\left[v_x(t,X_t^*,Z_t)\piv^*_t 
	+ v_z(t,X_t^*,Z_t)
	\begin{pmatrix}
		1\\
		-1
	\end{pmatrix}
	\right]^\top
	\Sig\,\dd \Wv_t\\
	&=\Big[\frac{v(t,X_t^*,Z_t)}{\frac{\gam}{\gam-1}(X_t-x^*)}(X^*_t-x^*) \alv^*(t,Z_t)\\
	&\qquad\qquad+ v(t,X_t^*,Z_t)(g(t)+h(t)Z_t)
	\begin{pmatrix}
		1\\
		-1
	\end{pmatrix}
	\Big]^\top
	\Sig\,
	\dd \Wv_t\\
	&= v(t,X_t^*,Z_t)\left[\left(1-\frac{1}{\gam}\right)\alv^*(t,Z_t) 
	+ \big(g(t)+h(t)Z_t\big)
	\begin{pmatrix}
		1\\
		-1
	\end{pmatrix}
	\right]^\top
	\Sig\,
	\dd \Wv_t\\
	& = Y_t \left(\alvt(t) + \betavt(t) Z_t\right) \dd \Wv_t,
\end{align}
for $t\in[0,T]$, where we have defined $\alvt(t)$ and $\betavt(t)$ such that
\begin{align}
	\alvt(t) + \betavt(t) z = \left(1-\frac{1}{\gam}\right)\alv^*(t,z) 
	+ \big(g(t)+h(t)z\big)
	\begin{pmatrix}
		1\\
		-1
	\end{pmatrix};\quad 0\le t\le T, z\in\Rb.
\end{align}
Note, in particular, that $\alvt(t)$ and $\betavt(t)$ are bounded on $[0,T]$ since $g(t)$ and $h(t)$ are continuous on $[0,T]$. We can then repeat the argument in the proof of Proposition \ref{prop:NA} in Appendix \ref{app:NA} to show that, for the time change $\altt(t)$ in \eqref{eq:timechange}, the time-changed processes $(\Yt_\altt := Y_{t(\altt)})_{0\le\altt\le \Tt}$ is a martingale. Thus, $(Y_t:= v(t, X^*_t,Z_t)\big)_{0\le t\le T}$ is a martingale, as we set out to prove.


\section{Proof of Lemma \ref{lem:sign}}

We first show that, for $\gam>1$, $h(t) >0$ for all $t\in [0,T)$.  With  $\gam>1$, ODE \eqref{eq:h_ode-1st} yields
\begin{align}
	h'(T)=\frac{(1-\gam)\kap^2}{\eps^2(1-\rho^2)}<0.
\end{align}
Since $h(T)=0$, we must have $h(t)>0$ in a left neighborhood of $T$, say $t\in[s,T)$ for some $0\le s<T$. For $h(t)>0$ for all $t\in[0,T)$, we show that $h(t)\ne0$ for all $t\in[0,s)$.

Assume otherwise, that is, $h(t_0)=0$ for some $t_0\in[0,s)$. Then, it follows from \eqref{eq:h_ode-1st} that 
\begin{align}
	h'(t_0)=\frac{(1-\gam)\kap^2}{(1-\rho^2)(T+\eps-t_0)^2}<0.
\end{align}
This implies  that $h(t)<0$ in a right neighborhood of $t_0$, say $t\in(t_0,s']$, for some $s'<s$. Now, since the smooth function $h$ increases from $h(s')<0$ to $h(s)>0$ on the interval $[s',s]$, there must be a point $\tau\in(s',s)$ such that $h(\tau) = 0$ and $h'(\tau)\ge0$. However, by \eqref{eq:h_ode-1st}, 
\begin{align}
	h'(\tau)=\frac{(1-\gam)\kap^2}{(1-\rho^2)(T+\eps-\tau)^2}<0,
\end{align}
which is a contradiction. Therefore, there cannot be such $t_0$, as desired. The proof for the case with  $0\le\gam<1$ follows from a similar argument and is thus omitted.
 

\begin{thebibliography}{}

\bibitem[\protect\citeauthoryear{Adjemian, Garcia, Irwin, and Smith}{Adjemian
  et~al.}{2013}]{adjemian2013}
Adjemian, M.~K., P.~Garcia, S.~Irwin, and A.~Smith (2013).
\newblock Non-convergence in domestic commodity futures markets: Causes,
  consequences, and remedies.
\newblock {\em US Department of Agriculture, Economic Research Service\/}~{\em
  115}, 155381.

\bibitem[\protect\citeauthoryear{Brennan and Schwartz}{Brennan and
  Schwartz}{1988}]{BrennanSchwartz1988}
Brennan, M.~J. and E.~S. Schwartz (1988).
\newblock Optimal arbitrage strategies under basis variability.
\newblock In M.~Sarnat (Ed.), {\em Essays in Financial Economics}. North
  Holland.

\bibitem[\protect\citeauthoryear{Brennan and Schwartz}{Brennan and
  Schwartz}{1990}]{BrennanSchwartz1990}
Brennan, M.~J. and E.~S. Schwartz (1990).
\newblock Arbitrage in stock index futures.
\newblock {\em The Journal of Business\/}~{\em 63\/}(1), S7--S31.

\bibitem[\protect\citeauthoryear{Brody, Hughston, and Macrina}{Brody
  et~al.}{2008}]{brody2008information}
Brody, D.~C., L.~P. Hughston, and A.~Macrina (2008).
\newblock Information-based asset pricing.
\newblock {\em International Journal of Theoretical and Applied Finance\/}~{\em
  11\/}(01), 107--142.

\bibitem[\protect\citeauthoryear{Bulthuis, Concha, Leung, and Ward}{Bulthuis
  et~al.}{2017}]{bulthuis2016optimal}
Bulthuis, B., J.~Concha, T.~Leung, and B.~Ward (2017).
\newblock Optimal execution of limit and market orders with trade director,
  speed limiter, and fill uncertainty.
\newblock {\em International Journal of Financial Engineering\/}~{\em
  4\/}(2-3), 1750020.

\bibitem[\protect\citeauthoryear{Cartea, Gan, and Jaimungal}{Cartea
  et~al.}{2018}]{CarteaGanJaimungal2018}
Cartea, {\'A}., L.~Gan, and S.~Jaimungal (2018).
\newblock Trading cointegrated assets with price impact.
\newblock Mathematical Finance, Forthcoming 2018, arXiv:1807.01428 [q-fin.TR].

\bibitem[\protect\citeauthoryear{Cartea and Jaimungal}{Cartea and
  Jaimungal}{2016}]{cartea2016algorithmic}
Cartea, {\'A}. and S.~Jaimungal (2016).
\newblock Algorithmic trading of co-integrated assets.
\newblock {\em International Journal of Theoretical and Applied Finance\/}~{\em
  19\/}(06), 1650038.

\bibitem[\protect\citeauthoryear{Cartea, Jaimungal, and Kinzebulatov}{Cartea
  et~al.}{2016}]{cartea2016algorithmic2}
Cartea, {\'A}., S.~Jaimungal, and D.~Kinzebulatov (2016).
\newblock Algorithmic trading with learning.
\newblock {\em International Journal of Theoretical and Applied Finance\/}~{\em
  19\/}(4), 1650028.

\bibitem[\protect\citeauthoryear{Chan}{Chan}{1992}]{chan1992further}
Chan, K. (1992).
\newblock A further analysis of the lead--lag relationship between the cash
  market and stock index futures market.
\newblock {\em The Review of Financial Studies\/}~{\em 5\/}(1), 123--152.

\bibitem[\protect\citeauthoryear{Cheridito, Filipovi\'{c}, and Yor}{Cheridito
  et~al.}{2005}]{CheriditoFilipovicYor2005}
Cheridito, P., D.~Filipovi\'{c}, and M.~Yor (2005).
\newblock Equivalent and absolutely continuous measure changes for
  jump-diffusion processes.
\newblock {\em Ann. Appl. Probab.\/}~{\em 15\/}(3), 1713--1732.

\bibitem[\protect\citeauthoryear{Chiu and Wong}{Chiu and
  Wong}{2011}]{ChiuWong2011}
Chiu, M. and H.~Wong (2011).
\newblock Mean-variance portfolio selection of cointegrated assets.
\newblock {\em Journal of Economic Dynamics and Control\/}~{\em 35},
  1369--1385.

\bibitem[\protect\citeauthoryear{Cox, Ingersoll, and Ross}{Cox
  et~al.}{1981}]{CoxIngersollRoss1981}
Cox, J.~C., J.~F. Ingersoll, and S.~A. Ross (1981).
\newblock The relation between forward and futures price.
\newblock {\em Journal of Financial Economics\/}~{\em 9\/}(December), 321--346.

\bibitem[\protect\citeauthoryear{Dai, Zhong, and Kwok}{Dai
  et~al.}{2011}]{dai2011optimal}
Dai, M., Y.~Zhong, and Y.~K. Kwok (2011).
\newblock Optimal arbitrage strategies on stock index futures under position
  limits.
\newblock {\em Journal of Futures markets\/}~{\em 31\/}(4), 394--406.

\bibitem[\protect\citeauthoryear{Fleming and Soner}{Fleming and
  Soner}{2006}]{FlemingSoner2006}
Fleming, W.~H. and H.~M. Soner (2006).
\newblock {\em Controlled Markov Processes and Viscosity Solutions}, Volume~25.
\newblock springer New York.

\bibitem[\protect\citeauthoryear{Garcia, Irwin, and Smith}{Garcia
  et~al.}{2015}]{Garcia2015}
Garcia, P., S.~H. Irwin, and A.~Smith (2015).
\newblock Futures market failure?
\newblock {\em American Journal of Agricultural Economics\/}~{\em 97\/}(1),
  40--64.

\bibitem[\protect\citeauthoryear{Irwin, Garcia, Good, and Kunda}{Irwin
  et~al.}{2011}]{irwin2011}
Irwin, S.~H., P.~Garcia, D.~L. Good, and E.~L. Kunda (2011).
\newblock Spreads and non-convergence in chicago board of trade corn, soybean,
  and wheat futures: Are index funds to blame?
\newblock {\em Applied Economic Perspectives and Policy\/}~{\em 33\/}(1),
  116--142.

\bibitem[\protect\citeauthoryear{Karatzas and Shreve}{Karatzas and
  Shreve}{1991}]{KartzasShreve1991}
Karatzas, I. and S.~Shreve (1991).
\newblock {\em {B}rownian Motion and Stochastic Calculus}.
\newblock Springer-Verlag.

\bibitem[\protect\citeauthoryear{Karatzas and Shreve}{Karatzas and
  Shreve}{1998}]{karatzas1998methods}
Karatzas, I. and S.~Shreve (1998).
\newblock {\em Methods of Mathematical Finance}, Volume~39.
\newblock Springer.

\bibitem[\protect\citeauthoryear{Kawaller, Koch, and Koch}{Kawaller
  et~al.}{1987}]{kawaller1987temporal}
Kawaller, I.~G., P.~D. Koch, and T.~W. Koch (1987).
\newblock The temporal price relationship between {S}\&{P} 500 futures and the
  {S}\&{P} 500 index.
\newblock {\em The Journal of Finance\/}~{\em 42\/}(5), 1309--1329.

\bibitem[\protect\citeauthoryear{Kim and Omberg}{Kim and
  Omberg}{1996}]{KimOmberg1996}
Kim, S. and E.~Omberg (1996).
\newblock Dynamic nonmyopic portfolio behavior.
\newblock {\em The Review of Financial Studies\/}~{\em 9\/}(1), 141--161.

\bibitem[\protect\citeauthoryear{Kitapbayev and Leung}{Kitapbayev and
  Leung}{2018}]{KitapLeung2017}
Kitapbayev, Y. and T.~Leung (2018).
\newblock Optimal mean-reverting spread trading: Nonlinear integral equation
  approach.
\newblock {\em Annals of Finance\/}~{\em 13\/}(2), 181--203.

\bibitem[\protect\citeauthoryear{Korn and Kraft}{Korn and
  Kraft}{2004}]{KornKraft2004}
Korn, R. and H.~Kraft (2004).
\newblock On the stability of continuous-time portfolio problems with
  stochastic opportunity set.
\newblock {\em Mathematical Finance\/}~{\em 14\/}(3), 403--414.

\bibitem[\protect\citeauthoryear{Lee and Papanicolaou}{Lee and
  Papanicolaou}{2016}]{lee2016pairs}
Lee, S. and A.~Papanicolaou (2016).
\newblock Pairs trading of two assets with uncertainty in co-integration's
  level of mean reversion.
\newblock {\em International Journal of Theoretical and Applied Finance\/}~{\em
  19\/}(08), 1650054.

\bibitem[\protect\citeauthoryear{Leung, Li, and Li}{Leung
  et~al.}{2018}]{leung2017optimal}
Leung, T., J.~Li, and X.~Li (2018).
\newblock Optimal timing to trade along a randomized {B}rownian bridge.
\newblock {\em International Journal of Financial Studies\/}~{\em 6\/}(3), 75.

\bibitem[\protect\citeauthoryear{Leung and Li}{Leung and
  Li}{2016}]{LeungXin2016}
Leung, T. and X.~Li (2016).
\newblock {\em Optimal Mean Reversion Trading: Mathematical Analysis And
  Practical Applications}.
\newblock World Scientific.

\bibitem[\protect\citeauthoryear{Leung and Ward}{Leung and
  Ward}{2015}]{LeungWard}
Leung, T. and B.~Ward (2015).
\newblock The golden target: Analyzing the tracking performance of leveraged
  gold {ETF}s.
\newblock {\em Studies in Economics and Finance\/}~{\em 32\/}(3), 278--297.

\bibitem[\protect\citeauthoryear{Leung and Yan}{Leung and
  Yan}{2018}]{LeungYan2018}
Leung, T. and R.~Yan (2018).
\newblock Optimal dynamic pairs trading of futures under a two-factor
  mean-reverting model.
\newblock {\em International Journal of Financial Engineering\/}~{\em 5\/}(3),
  1850027.

\bibitem[\protect\citeauthoryear{Leung and Yan}{Leung and
  Yan}{2019}]{LeungYan2019}
Leung, T. and R.~Yan (2019).
\newblock A stochastic control approach to managed futures portfolios.
\newblock {\em International Journal of Financial Engineering\/}~{\em 6\/}(1),
  1950005.

\bibitem[\protect\citeauthoryear{Liu and Longstaff}{Liu and
  Longstaff}{2004}]{LiuLongstaff2004}
Liu, J. and F.~A. Longstaff (2004).
\newblock Losing money on arbitrage: Optimal dynamic portfolio choice in
  markets with arbitrage opportunities.
\newblock {\em The Review of Financial Studies\/}~{\em 17\/}(3), 611--641.

\bibitem[\protect\citeauthoryear{Liu and Timmermann}{Liu and
  Timmermann}{2013}]{LiuTimmermann2012}
Liu, J. and A.~Timmermann (2013).
\newblock Optimal convergence trade strategies.
\newblock {\em Review of Financial Studies\/}~{\em 26\/}(4), 1048--1086.

\bibitem[\protect\citeauthoryear{Merton}{Merton}{1969}]{Merton1969}
Merton, R.~C. (1969).
\newblock Lifetime portfolio selection under uncertainty: {T}he continuous-time
  case.
\newblock {\em Review of Economics and Statistics\/}~{\em 51\/}(3), 247--257.

\bibitem[\protect\citeauthoryear{Modest and Sundaresan}{Modest and
  Sundaresan}{1983}]{Modest1983}
Modest, D.~M. and M.~Sundaresan (1983).
\newblock The relationship between spot and futures prices in stock index
  futures markets: Some preliminary evidence.
\newblock {\em Journal of Futures Markets\/}~{\em 3\/}(1), 15--41.

\bibitem[\protect\citeauthoryear{Mudchanatongsuk, Primbs, and
  Wong}{Mudchanatongsuk et~al.}{2008}]{Primbsetal2008}
Mudchanatongsuk, S., J.~Primbs, and W.~Wong (2008).
\newblock Optimal pairs trading: a stochastic control approach.
\newblock In {\em Proceedings of the American Control Conference}, Seattle,
  Washington, pp.\  1035--1039.

\bibitem[\protect\citeauthoryear{Stoll and Whaley}{Stoll and
  Whaley}{1990}]{stoll1990dynamics}
Stoll, H.~R. and R.~E. Whaley (1990).
\newblock The dynamics of stock index and stock index futures returns.
\newblock {\em Journal of Financial and Quantitative analysis\/}~{\em 25\/}(4),
  441--468.

\bibitem[\protect\citeauthoryear{Tourin and Yan}{Tourin and
  Yan}{2013}]{tourin2013dynamic}
Tourin, A. and R.~Yan (2013).
\newblock Dynamic pairs trading using the stochastic control approach.
\newblock {\em Journal of Economic Dynamics and Control\/}~{\em 37\/}(10),
  1972--1981.

\bibitem[\protect\citeauthoryear{Zariphopoulou}{Zariphopoulou}{2001}]{zariphopoulou2001solution}
Zariphopoulou, T. (2001).
\newblock A solution approach to valuation with unhedgeable risks.
\newblock {\em Finance and stochastics\/}~{\em 5\/}(1), 61--82.

\end{thebibliography}
\end{document}